\documentclass[11pt,reqno]{amsart}

\usepackage[utf8]{inputenc}
\usepackage{amsfonts,latexsym,amssymb,amsmath,amsthm,amsrefs,etex,slashed,cancel}
\usepackage{todonotes}
\usepackage{hyperref}
\usepackage{a4wide}
\usepackage{enumerate}
\usepackage{MnSymbol}
\usepackage{nicefrac}
\usepackage{graphicx}

\renewcommand{\Re}{\operatorname{Re}}
\renewcommand{\Im}{\operatorname{Im}}
\newcommand{\dom}{\mathrm{odom}}
\DeclareMathOperator{\supp}{supp}
\newcommand{\der}{\mathrm{d}}
\newcommand{\rmi}{\mathrm{i}}

\newcommand{\calO}{\mathcal{O}}

\renewcommand{\dom}{\mathrm{dom}}

\newcommand{\U}{\mathcal{U}}

\newcommand{\ret}{\mathrm{ret}}
\newcommand{\av}{\mathrm{av}}

\newcommand{\R}{\mathbb{R}}
\newcommand{\C}{\mathbb{C}}
\newcommand{\calA}{\mathcal{A}}
\newcommand{\calR}{\mathcal{R}}
\newcommand{\calE}{\mathcal{E}}

\newcommand{\calH}{\mathcal{H}}
\newcommand{\calL}{\mathcal{L}}
\newcommand{\TT}{T}
\newcommand{\WF}{\mathrm{WF}}
\newcommand{\WFa}{\mathrm{WF}_a}
\newcommand{\V}{\mathcal{V}}
\newcommand{\muS}{\mu\mathrm{S}}
\newcommand{\calU}{\mathcal{U}}

\newtheorem{theorem}{Theorem}[section]
\newtheorem{definition}[theorem]{Definition}
\newtheorem{example}[theorem]{Example}
\newtheorem{lemma}[theorem]{Lemma}
\newtheorem{assumption}[theorem]{Assumption}

\newtheorem{proposition}[theorem]{Proposition}

\newtheorem{rem}[theorem]{Remark}

\title[Analytic states in QFT]{Analytic states in quantum field theory on curved spacetimes}

\author[A.~Strohmaier]{Alexander Strohmaier}
\address{Leibniz University Hannover, Institute of Analysis, 30167 Hannover, Germany}  \email{a.strohmaier@math.uni-hannover.de} 

\author[E. Witten]{Edward Witten}
\address{Institute for Advanced Study,
School of Natural Sciences,
One Einstein Drive,
Princeton, NJ 08540}
\email{witten@ias.edu}

\thanks{Research of EW supported in part by NSF Grant PHY-2207584.}

\begin{document}
%

\begin{abstract}
 We discuss high energy properties of states for (possibly interacting) quantum fields in curved spacetimes. In particular, if the spacetime is real analytic, we show that an analogue of the timelike tube theorem and the Reeh-Schlieder property hold with respect to states satisfying a weak form of microlocal analyticity condition. The former means the von Neumann algebra of observables of a spacelike tube equals the von Neumann algebra of observables of a significantly bigger region, that is obtained by deforming the boundary of the tube in a timelike manner. This generalizes theorems by Borchers and Araki (\cite{MR119866,MR153340}) to curved spacetimes.  \end{abstract}

\maketitle

\vspace{.3cm}
\section{Introduction}

A quantum field on $d$-dimensional Minkowski spacetime can be thought of as an operator-valued distribution $\Phi$ on $\R^d$. For a quantum field $\Phi$ acting on a dense domain $D \subset \calH$ in a Hilbert space $\calH$ with vacuum vector $\Omega$ the spectrum condition implies that in case a test function $f$ has Fourier transform $\hat f$ supported away from the backward light cone, the smeared out field $\Phi(f)$ annihilates the vacuum, i.e. $\Phi(f) \Omega=0$. More generally, if $f_1,\ldots,f_n$ are test functions whose Fourier transforms $\hat f_1,\ldots,\hat f_n$ have the property that
$\supp \hat f_1 + \ldots + \supp \hat f_n$ does not intersect the backward light cone, we again have
$\Phi(f_1) \cdots \Phi(f_n) \Omega=0$. This is a manifestation of the fact that only states with non-negative energy can be created out of the vacuum, hence test functions that do not contain a positive energy component yield an operator that annihilates the vacuum vector.

In a general curved spacetime $(M,g)$, due to the absence of a translational symmetry, there is no notion of energy or the Fourier transform. Positivity of energy in the above sense can only  be expected in an asymptotic sense. To explain this we consider families of test functions $f_h$ that depend on a parameter $h \in (0,1]$. As $h \searrow 0$ these test functions need not converge, but they will have localisation properties in phase space that transform covariantly under a change of coordinates. In this paper we are interested in the analytic category, i.e. we assume that the spacetime $(M,g)$ has a real analytic set of coordinates with respect to which the metric $g$ is real analytic. Hence, we only need to consider analytic coordinate changes. The notion of microsupport of the family $f_h$ captures the localisation properties of the family of functions $f_h$ as $h \searrow 0$ in a manifestly covariant way. Roughly speaking the microsupport is the set of elements $(x,\xi)$ in the cotangent space $T^*M$ with the property that 
the inner product 
$$
 \langle\psi_{x,\xi,h}, f_h \rangle =  (\pi h)^{-\frac{d}{4}} \int   e^{-\frac{(x-y)^2}{2h}} e^{\frac{\rmi}{h} (x-y) \xi} f_h(y) \der x
$$
in local coordinates with a coherent state $\psi_{x,\xi,h}(y) = (\pi h)^{-\frac{d}{4}} e^{-\frac{1}{2 h}(x-y)^2} e^{-\frac{\rmi}{h} (x-y) \cdot \xi}$ localized at the point $(x,\xi)$
is not exponentially small as $h \searrow 0$ (see Section \ref{microsupport} for the precise mathematical definition). 
Here we say a function is exponentially small if it is of order $e^{-\delta h^{-1}}$ for some $\delta>0$ in some appropriate norm or semi-norm. The observation at the heart of microlocal analysis in the analytic category is that this notion is covariant under analytic change of coordinates, and it therefore makes sense on an analytic curved spacetime. The microsupport in this way becomes a subset of the cotangent bundle.
Similar notions exist in the category of smooth manifolds, but the exponential decay needs to be replaced by
decay faster than any power of $h$.
{The role of $h$ is that of a scaling parameter. As $h \searrow 0$ the wave packet $\psi_{x,\xi,h}$ localizes at the phase space point $(x,\xi)$ at a scale $h^{\frac{1}{2}}$ after a rescaling $(x,\xi) \mapsto (x,h \xi)$ has been applied. The physical intuition is that this corresponds to the regime of high energy in which asymptotic localization in phase space is possible. This is ultimately the reason why the corresponding notions make sense on manifolds.
}

For a quantum field in a curved spacetime it is therefore natural to hypothesize the existence of states $\Omega$ for which 
\begin{align*}
 \Phi(f_{1,h}) \cdots  \Phi(f_{n,h}) \Omega
\end{align*}
is exponentially small as soon the microsupport of the family $f_{1,h}(x_1) \cdots f_{n,h}(x_n)$ of functions on $M \times \ldots \times M$ is contained only in a set that does not correspond to an allowed physical process.
To be more precise, if the microsupport of $f_{1,h}$ is contained in a region of phase space that cannot be reached by classical scattering processes from the microsupports of 
$f_{k,h}, k=2,\ldots,n$ the vector $\Phi(f_{1,h}) \cdots  \Phi(f_{n,h}) \Omega$ should be exponentially small. 
In scattering theory the set of admissible processes is restricted by Landau diagrams (see \cite{MR992456} and \cite{MR0367657})
and one expects similar restrictions for correlations in curved spacetimes. For example it has been conjectured  \cite{MR3636845} that for certain types of quantum fields in curved spacetimes analytic singularities only propagate along null-lines and scatter classically. In the generality of the Wightman framework this seems to be too strong  and excludes various generalized free fields. Another condition that is more likely to include general Wightman fields in Minkowski space is that analytic singularities propagate along causal curves and scatter classically. Other even weaker  propositions have been stated in \cite{BFK}.

In this paper we are concerned with a very weak form of such a physical condition. Namely, we call a state $\Omega$
{\sl analytic} if $\Phi(f_{1,h}) \cdots  \Phi(f_{n,h}) \Omega$ is exponentially small as soon as the microsupport of 
$f_{1,h}$ has positive distance from the backward light cone, and the microsupports of the $f_{k,h}, k=2,\ldots,n$ are contained in the zero section of the cotangent bundle in a uniform manner. 
Roughly, it is not possible asymptotically to extract energy from such a state.
The name ``analytic state" is motivated by analogy from properties of analytic vectors with respect to group actions.

We show that for such states we have the timelike tube theorem and the Reeh-Schlieder theorem. We give here informal versions of the two theorems, {assuming there exists a cyclic analytic vector}.

{\bf Timelike tube theorem,} (Theorem \ref{timeliketube}): Let $\mathcal{R}(\calO)$ be the local von-Neumann algebra associated to the space-time region $\calO \subset M$. Let $\calO_T$ be a larger spacetime region obtained from $\calO$ by deforming timelike curves 
in $\calO$ in a timelike manner with fixed endpoints. Then $\mathcal{R}(\calO) = \mathcal{R}(\calO_T)$.\\

In case $\calO$ is an infinite timelike tube then the set $\calO_T$ is in some cases the entire spacetime.

{\bf Reeh-Schlieder theorem,} (Theorem \ref{Reeh-Schlieder}): $\mathcal{R}(\calO) \Omega$ is dense in the Hilbert space for any non-empty open subset $\calO \subset M$.

{
Both theorems are of structural importance in rigorous treatments in algebraic quantum field theory. The Reeh-Schlieder theorem is the basic result governing entanglement in quantum field theory, that opened the path for the use of modular theory.
The timelike tube theorem shows that in defining a net of algebras associated to open sets, only open sets that are their
own timelike envelopes need to be considered.  Combining this with relativistic causality leads to further conditions, as originally discussed by Araki (\cite{MR153340}).
}
We state and prove both of the theorems as they are similar in nature, but note that the validity of the Reeh-Schlieder theorem in this context has been known under similar conditions (\cite{SVW}). To our knowledge the validity of the timelike tube theorem has not been noted in this general context that is thought to include interacting fields.

 Whether or not the timelike tube theorem holds for even the free Klein-Gordon field on a general non-necessarily  analytic globally hyperbolic spacetime remains an open question. A counterexample by Alinhac and Baouendi (\cite{MR1363855}) shows that there exists a smooth complex-valued potential $V \in C^\infty(\R^3,\C)$ on three-dimensional Minkowski space, and a smooth function $u \in C^\infty(\R^3,\C)$ with support  $\mathrm{supp}(u)$ equal to the half-plane $\{(t,x,y) \in \R^3 \mid x>0\}$ such that $\Box u + V u =0$ in an open neighborhood of zero. It is currently unknown if such an example can also be constructed for the metric wave operator or a real potential. Under a partial analyticity assumption the case of the Klein-Gordon field was treated in \cite{MR1846957}.

{
It is likely that all of our analysis carried out in the analytic category works also if one replaces the class of analytic functions with the quasi-analytic Denjoy-Carleman class which is slightly more general. A great deal of the geometric framework for this class has indeed been worked out in \cite{Furdos2018GeometricMA}.}

The article is organized as follows. In Section \ref{FBIsection} we discuss various notions of microlocal analysis such as the FBI (Fourier-Bros-Iagolnitzer) transform, the microsupport, and the analytic wavefront set. We briefly explain in Section \ref{seclytic} how these notions come up naturally in quantum theory and link with the notion of analytic vectors for the time-evolution operator.
 In Section \ref{qftsec} we give the mathematical framework for treating quantum fields in curved spacetimes and introduce and discuss the notions of analyticity and tempered analyticity. Section \ref{ttt} contains the statement and proof of the timelike tube and the Reeh-Schlieder theorem. The appendices provide details of the mathematical framework.

\section{Microlocal analysis and the FBI transform} \label{FBIsection}

Schr\"odinger quantum mechanics of a single particle on $\R^d$ is described by the Hilbert space $\mathcal{H} = L^2(\R^d)$ and the Schr\"odinger time-evolution $U(t)$ which is explicitly determined by the Schr\"odinger equation $U(t) = \exp(-\rmi t H)$, where $H$ is a self-adjoint operator, the Hamiltonian. The Hilbert space is concretely realized here as space of square integrable functions which encodes the quantum mechanical interpretation of a state $\Psi \in \calH$, namely $| \Psi(x) |^2 \der x$ is a probability measure that describes the probability distribution measuring the presence of the particle in a space-region. The above is the position representation of the Hilbert space.
Using the Fourier transform 
$\mathcal{F} : L^2(\R^d) \to L^2(\R^d), f \mapsto \hat f,
\hat f(\xi) = (2 \pi)^{-\frac{d}{2}} \int_{\R^d} f(x) e^{-\rmi x \cdot \xi} \der x$ one can pass to the momentum representation. The probability distribution $| \hat \Psi(\xi) |^2 \der \xi$ describes the probability distribution measuring the momentum of a particle in a region. Whereas precise localisation in phase space is not possible due to the Heisenberg uncertainty relation there are still various asymptotic notions of phase-space localisation. An example are the coherent states
$$
 \psi_{x,\xi,h}(y)=(\pi h)^{-\frac{d}{4}}  e^{-\frac{(y-x)^2}{2h}} e^{-\frac{\rmi}{h} (x-y) \xi},
$$
which localize in a rescaled version of phase space at the point $(x,\xi)$ as $h \searrow 0$. Here $h \in (0,\infty)$ serves as formal parameter which should not be confused with Planck's constant. The idea of asymptotic localisation in phase space has had a profound influence on the theory of partial differential equations and the corresponding analysis is more commonly referred to as microlocal analysis.  It can be used to describe regularity properties of distributions in phase space. There are various approaches to microlocal analysis which we review and link in the appendix. We will focus here on an approach based on semi-classical analysis and the FBI transform, which is motivated by developments in the theory of partial differential equations.
Remarkably the original ideas linking analyticity and the FBI transform were in parts discovered in the analysis of the analytic properties of the scattering matrix in quantum field theory (see e.g. \cite{MR0367657}).

As usual we let $\mathcal{S}(\R^d)$ be the space of Schwartz functions. Its topological dual is the space $\mathcal{S}'(\R^d)$ of tempered distributions. As above it is convenient to let transforms depend on an extra parameter $h \in (0,\infty)$. One defines the semi-classical Fourier transform $\mathcal{F}_h$, by
\begin{equation}
 (\mathcal{F}_h f)(\xi) = \frac{1}{(2 \pi h)^\frac{d}{2}} \int_{\R^d} f(x) e^{-\frac{\rmi}{h} x \cdot \xi} \der x.
\end{equation}
For fixed $h$ this defines a continuous map $\mathcal{F}_h: \mathcal{S}(\R^d) \to \mathcal{S}(\R^d)$ that continuously extends
as $\mathcal{F}_h: \mathcal{S}'(\R^d) \to \mathcal{S}'(\R^d)$ and a unitary map $\mathcal{F}_h: L^2(\R^d) \to L^2(\R^d)$.
Given a Schwartz function $f \in \mathcal{S}(\R^d)$, and a number $h \in (0,\infty)$, one defines the semi-classical FBI-transform $(\TT_h f)(x,\xi)$ at the point $(x,\xi) \in \R^d \times \R^d$ as
\begin{equation}
 (\TT_h f)(x,\xi) = \alpha_h \int_{\R^d} e^{-\frac{1}{2 h}(x-y)^2} e^{\frac{\rmi}{h} (x-y) \cdot \xi} f(y) \der y, \quad \alpha_h = 2^{-\frac{d}{2}} (\pi h)^{-\frac{3d}{4}}.
\end{equation}
The FBI transform defines for every $h>0$ an isometry $L^2(\R^d) \to L^2(\R^{d} \times \R^d)$. We also have the continuous mapping properties
\begin{align*}
 \TT_h: \mathcal{S}(\R^d) \to \mathcal{S}(\R^d \times \R^d),\quad \TT_h:\mathcal{S}'(\R^d) \to \mathcal{S}'(\R^d \times \R^d),\\
  \TT_h^*:  \mathcal{S}(\R^d \times \R^d) \to \mathcal{S}(\R^d) ,\quad \TT_h^*: \mathcal{S}'(\R^d \times \R^d) \to \mathcal{S}'(\R^d) ,
 \end{align*}
 for fixed $h$. We have $T_h^* T_h = \mathrm{id}$ on $L^2(\R^d), \mathcal{S}(\R^d)$, and  $\mathcal{S}'(\R^d)$ respectively. Moreover, one checks directly that $e^\frac{\xi^2}{2 h} (T_h f)(x,\xi)$ is holomorphic in $z= x- \rmi \xi$.
 
This transform is closely related to the Bargman transform and in fact achieves an expansion of the state into coherent states. Indeed,
$$
 \psi_{x,\xi,h}(y) = (\pi h)^{-\frac{d}{4}} e^{-\frac{1}{2 h}(x-y)^2} e^{\frac{\rmi}{h} (y-x) \cdot \xi},
$$
is an $L^2$-normalized coherent state centered at the point $(x,\xi)$ in phase space. For fixed $h>0$, the formula $\TT_h^* \TT_h=1$ can be expressed as
\begin{equation} \label{csrep}
u(y) = (2 \pi h)^{-\frac{d}{2}}\int_{\R^{2d}} (\TT_h u)(x,\xi) \psi_{x,\xi,h}(y) \der x \der \xi.
\end{equation}
In other words, the FBI-transform allows to superpose the distribution from a family of coherent states.  Another inversion formula that is useful in this context is
\begin{equation}\label{inversion}
 u(x) = 2^{-\frac{d}{2}-1}(\pi)^{-\frac{d}{4}} \int_{|\xi|=1} \int_0^\infty  h^{-1-\frac{d}{4}} \left( 1 + \xi \cdot \frac{h}{\rmi} \mathrm{grad}_x \right) (\TT_{h} u)(x,\xi) \der h\, \der \xi,
\end{equation} 
the proof of which can be found in \cite{MR1065136}*{Sect. 9.6}.

As $h \searrow 0$ the function $e^{-\frac{1}{2h }(x-y)^2+\frac{\rmi}{h} (y-x) \cdot \xi}$ localizes more and more at the phase space point $(y,\xi)$. If a distribution $u \in S'(\R^d)$ vanishes near the point $y$ in fact the FBI-transform is exponentially decaying in $h^{-1}$ as $h \searrow 0$. 
It is this localisation property that makes it useful to study local properties of distributions.
The FBI transform is also essentially the same as the wavepacket transform of \cite{MR507783} introduced to study mapping properties of Fourier integral operators.

\subsection{Families of test functions and distributions}

It will be convenient to work with general $h$-dependent families of test functions and distributions.
In case $(f_h)$ is a family of Schwartz functions we say that this family is {\sl polynomially bounded} if for every Schwartz semi-norm $p$ we have that $p(f_h) = O(h^{-N})$ as $h \searrow 0$ for some $N>0$ potentially depending on the semi-norm.
We call a family $(u_h)$ of tempered distributions $u_h \in \mathcal{S}'(\R^d)$ polynomially bounded if there exists a  continuous semi-norm $p: \mathcal{S}(\R^d) \to [0,\infty)$ such that for some $N>0$ we have $\| u_h(f)\| \leq h^{-N} p(f)$ for all $f \in \mathcal{S}(\R^d), h \in (0,1]$. 

{
 The notions discussed here such as polynomial boundedness for families of distributions carry over to families of distributions $(u_h), u_h \in \mathcal{S}'(\R^d,V)$ taking values in a Banach space $V$. We do not write this out explicitly in an attempt to not overload the notation, but we write  $\| u(f) \|$ to remind the reader that these definitions carry over to Banach space-valued families of distributions, where the norm  $\| \cdot \|$ is the norm on $V$.
}

The two notions of polynomial boundedness are compatible as the pairing of a polynomially bounded family of distributions and a polynomially bounded family of test-functions is a polynomially bounded function of $h$. 

We will denote by $\mathcal{S}'_h(\R^d)$ the vector space of polynomially bounded families of tempered distributions and by $\mathcal{S}_h(\R^d)$ the algebra of polynomially bounded families of Schwartz functions.
{In the examples we have in mind the restriction to the class of polynomially bounded families is not a serious one.  The advantage of restricting to the class of polynomially bounded functions is that it is stable under forming tensor products and does not change exponential decay rates. We call a family of test function $(f_h) \in \mathcal{S}_h(\R^d)$ exponentially small with decay rate $\delta>0$ if for every Schwartz semi-norm $p$ we have $p(f_h) = O(e^{-\delta h^{-1}})$ as $h \searrow 0$.
For example, a sufficient condition for a family of the form $f_{1,h} \otimes \cdots \otimes f_{n,h} \in \mathcal{S}_h(\R^d \times \ldots \times \R^d)$ to be  exponentially small with decay rate $\delta>0$ is that one of the tensor factors is exponentially small with that decay rate. Similarly, an element in $u \in \mathcal{S}'_h(\R^d)$ paired with an exponentially small test function results in an exponentially small function of $h$.} 

The semi-classical Fourier transform $\mathcal{F}_h$ and the FBI transform $T_h$ can be viewed as maps
\begin{align*}
 &\TT_h: \mathcal{S}_h(\R^d) \to \mathcal{S}_h(\R^d \times \R^d),\quad \TT_h:\mathcal{S}_h'(\R^d) \to \mathcal{S}_h'(\R^d \times \R^d), \\
 &\mathcal{F}_h: \mathcal{S}_h(\R^d) \to \mathcal{S}_h(\R^d),\quad \mathcal{F}_h:\mathcal{S}_h'(\R^d) \to \mathcal{S}_h'(\R^d )
 \end{align*}
 {
 by applying the maps pointwise, i.e.
 \begin{equation*}
 (\TT_h f_h)_h(x,\xi) = \alpha_h \int_{\R^d} e^{-\frac{1}{2 h}(x-y)^2} e^{\frac{\rmi}{h} (x-y) \cdot \xi} f_h(y) \der y.
\end{equation*}
}

\subsection{Uniform microlocalisation of a family} \label{microsupport}

For a polynomially bounded family of test functions $(u_h) \in \mathcal{S}_h(\R^d)$ consider the FBI transform $(T_h u_h)$. This is again a polynomially bounded family of test functions in $\mathcal{S}_h(\R^{d}\times \R^d)$.
\begin{definition}
We say $(u_h) \in \mathcal{S}_h(\R^d)$ is {\sl microlocally uniformly exponentially small in $\calU \subseteq \R^{d} \times \R^d$} if there exists a $\delta>0$ such that for all $N>0$ we have the estimate
$$
 | ((1+ x^2 + \xi^2)^{N}T_h u_h)(x,\xi) | \leq C_N e^{-\delta h^{-1}} \textrm{ for all }  (x,\xi) \in \calU
$$
for some $C_N >0$. 
\end{definition}

This definition is very natural when working with Schwartz functions as test functions, in fact the decay requirements can be expressed in terms of Schwartz semi-norms, as we explain in Prop. \ref{swchar}.
The complementary notion is that of being uniformly microsupported in a set $K$, which essentially means we have microlocal uniform exponential smallness for all points that have a positive distance to $K$.

\begin{definition}
 We say that $(u_h) \in \mathcal{S}(\R^d)$ is {\sl uniformly microsupported} in $K \subset \R^d \times \R^d$
if it is microlocally uniformly exponentially small in $\calU_\epsilon = \{(x,\xi) \in \R^d \times \R^d \mid \mathrm{dist}((x,\xi),K) > \epsilon\}$ for all $\epsilon>0$. In this case we also say $u_h$ is uniformly microsupported away from the complement of $K$.
\end{definition}

This is intimately related to the notion of microsupport of a polynomially bounded family of distributions.
\begin{definition}
 The microsupport $\muS(u_h)$ of a polynomially bounded family of distributions $(u_h) \in \mathcal{S}'_h(\R^d)$ is the complement of the set of points  $(x_0,\xi_0) \in \R^d \times \R^d$ such that we have
 $$
 \| T_h u_h(x,\xi) \| < C e^{-\delta h^{-1}},
$$
uniformly for all $(x,\xi)$ near $(x_0,\xi_0)$. 
\end{definition}

The microsupport measures the (exponential) localisation properties of $(u_h)$ in phase space as $h \searrow 0$.
As an example, one can consider the family of coherent states
$$
 \psi_{x_0,\xi_0,h}(x) = (\pi h)^{-\frac{d}{4}} e^{-\frac{1}{2 h}(x-x_0)^2} e^{\frac{\rmi}{h} (x-x_0) \cdot \xi_0},
$$
which is a family of real analytic Schwartz functions. As an $h$-dependent family the microsupport is given by $\{(x_0,\xi_0)\}$ as the functions localize in phase space to this point as $h \searrow 0$.
Another example is $\chi(x)\psi_{x_0,\xi_0,h}(x)$ with $\xi_0=0$ and $\chi$ a compactly supported smooth test function that equals one near $x_0$.
As this example shows a family of smooth functions may also have non-zero microsupport, and we can therefore investigate how the field operator behaves on an $h$-dependent family $f_h$ of smooth compactly supported test functions.

The notion of being uniformly microsupported captures more than the notion of the microsupport. As an example let 
$$
 \psi_{x_h,\xi_h,h}(x) = (\pi h)^{-\frac{d}{4}} e^{-\frac{1}{2 h}(x-x_h)^2} e^{\frac{\rmi}{h} (x-x_h) \cdot \xi_h}
$$
be a family of coherent states with centered at a point $(x_h,\xi_h)$ that also depends on $h$ and escapes to $\infty$ at a fast enough rate as $h \to 0$. Then the microsupport of $\psi_{x_h,\xi_h,h}(x) + \psi_{x_0,\xi_0,h}(x)$ equals $\{(x_0,\xi_0)\}$ but the family is not microlocally uniformly exponentially small away in the complement of any ball $B_\epsilon((x_0,\xi_0))$ and is therefore not uniformly microsupported at $\{(x_0,\xi_0)\}$.

If we pass to the test function space $C^\infty_0(\R^d)$ there is an inherited notion of uniform microlocalisation that generalizes to real analytic manifolds. To make this precise we need to consider families of compactly supported $h$-dependent polynomially bounded test functions $f_h$. We do this for a general manifold $M$. We say a family $(f_h), f_h \in C^\infty_0(M)$
is polynomially bounded if
\begin{itemize}
\item there exists a compact subset $K \subset M$ such that $\mathrm{supp}(f_h) \subset K$ for all values of the parameter $h$,
\item for any smooth cut-off function $\chi \in C^\infty_0(M)$ compactly supported in a chart domain the family $(\chi f_h)$ is polynomially bounded in $\mathcal{S}(\R^d)$ with respect to the local coordinates of the chart.
\end{itemize}
The space of polynomially bounded families $(f_h)$ of test functions $f_h \in C^\infty_0(M)$ will be denoted by $C^\infty_{0,h}(M)$.
For such families of distributions the notion of being uniformly exponentially small in a subset of the cotangent bundle make sense. The reason is that by Prop. \ref{cochart} the notion of exponential smallness transforms covariantly under an analytic change of coordinates. A mildly subtle point is that upon localisation into a chart domain using a cut-off function $\chi$ we can expect uniform exponential smallness only where the cutoff function is constant, because outside this set it is not real
analytic. This leads to the following definition.

\begin{definition}
Given a compact set $K \subset M$ and a subset $N \subset \pi^{-1}(K) \subset T^*M$, we say 
$(f_h)$ is uniformly microlocally exponentially small in $N$ if for any analytic coordinate chart $\rho: M \supset \calO \to \rho(\calO) \subset\R^d$,
any relatively compact open set $\calU \subset \calO$ and any test function $\chi \in C^\infty_0(M)$ that equals one near $\calU$ the family $(\chi f_h) \circ \rho^{-1}$ is uniformly microlocally exponentially small in $(\rho^{-1})^*(\pi^{-1}(\calU) \cap N)$.
\end{definition}

Whereas this definition requires the function to be uniformly exponentially small with respect to any coordinate charts, it is sufficient to check this in a system of charts as long as the sets $\calU$ cover $K$, again by
 Prop. \ref{cochart}. By compactness it is therefore sufficient to check this in finitely many coordinate charts. The notion of uniform microlocalisation and microsupport for analytic manifolds and families in $C^\infty_{0,h}(M)$ is therefore consistent with the notion on $\R^d$ as above.
We note however that the requirement for all the test functions of the family to be supported in a fixed compact set is essential here.

\begin{proposition} \label{secondchar}
 Let $(u_h)\in \mathcal{S}'_h(\R^d)$ be a polynomially bounded family of tempered distributions, and let $(x,\xi) \in \R^d \times \R^d$.
 Then $(x,-\xi) \notin \muS(u_h)$ if and only if there exists an $\epsilon>0$ such that the following holds.
 For all families $(f_h) \in C^\infty_{0,h}(\R^d)$, uniformly microsupported in the ball $B_\epsilon(x,\xi)$ about the point $(x,\xi)$, we have
 $\|u_h(f_h)\| = O(e^{-\delta h^{-1}})$ for some $\delta>0$ as $h \searrow 0$.
\end{proposition}
\begin{proof}
 Let $\tilde \chi \in C^\infty_0(\R^d)$ be a cut-off function that equals one in a neighborhood of $x$.
 Replacing $u_h$ by $\tilde \chi u_h$ we can assume that $(u_h)$ is supported near the point $(x,-\xi)$.
 Note that by Prop. \ref{multmic}, if $(f_h)$ is uniformly microsupported at the point $(x,\xi)$, then so is $(\tilde \chi f_h)$.
 Assume first that $(x,-\xi) \notin \muS(u_h)$, or equivalently $(x,\xi) \notin \muS(\overline{u_h})$. Then, for some $\epsilon>0$ we have, uniformly for all $(y,\eta) \in B_{2\epsilon}(x,\xi)$,
the estimate
$\| T_h \overline{u_h}(y,\eta) \| \leq C e^{-\delta h^{-1}} $ for some $C,\delta>0$ and all $h \in (0,1]$. Suppose that $(f_h) \in C^\infty_{0,h}(\R^d)$ is microlocally uniformly exponentially small outside $B_{\epsilon}(x,\xi)$. We have
$$
 u_h(f_h) = \langle \chi_1  T_h (\overline{ u_h}), T_h f_h \rangle_{L^2(\R^{2d})} +  \langle T_h \overline{u_h}, \chi_2  T_h (f_h) \rangle_{L^2(\R^{2d})},  
$$
where $\chi_1,\chi_2$ are real valued cut-off functions with $\chi_1 + \chi_2 = 1$, with $\supp(\chi_1) \subset B_{2\epsilon}(x,\xi)$, and $\supp(\chi_2) \cap B_{\epsilon}(x,\xi) = \emptyset$. Then  $\chi_1  T_h (\overline{ u_h})$ is uniformly exponentially small and compactly supported, whereas  
$T_h f_h$ is polynomially bounded in $\mathcal{S}(\R^d \times \R^d)$. Hence, the first term is exponentially small. Similarly,
$T_h \overline{u_h}$ is polynomially bounded in $\mathcal{S}'(\R^d \times \R^d)$ and $\chi_2  T_h (f_h)$ is exponentially small in 
 $\mathcal{S}(\R^d \times \R^d)$. Therefore, also the second term is exponentially small. Note that this proof also works when $u_h$ takes values in a Hilbert space $\calH$, in which case the $L^2$-pairing above takes values in $\calH$.
  
 Now assume conversely, that for all families $(f_h) \in C^\infty_{0,h}(\R^d)$, uniformly microsupported in the ball $B_\epsilon(x,\xi)$ about the point $(x,\xi)$, we have
 $\|u_h(f_h)\| = O(e^{-\delta h^{-1}})$ for some $\delta>0$ as $h \searrow 0$.
We need to show that the FBI transform $f_h(y,\eta) = T_h(\overline{u_h})(y,\eta)$ of $\overline{u_h}$ is exponentially small 
 uniformly in $(y,\eta)$ in a neighborhood $\calU$ of $(x,\xi) \in \R^d \times \R^d$. Here $\overline{u_h}$ is the complex conjugate of the distribution $u_h$. Note that in case $u_h$ takes values in a Hilbert space $\overline{u_h}$ takes values in the complex conjugate Hilbert space.
 In $\overline{B_{\frac{\epsilon}{2}}(x,\xi)}$
 we can pick for each $h$ a point $(x_h,\xi_h)$ at which the maximum of $\| T_h \overline{u_h}(x,\xi)\|$ on 
 $\overline{B_{\frac{\epsilon}{2}}(x,\xi)}$ is attained. Now define $f_h(x) = \psi_{x_h,\xi_h,h}(x)$.
 Then 
 $$
   (T_h \overline{u_h})(x_h,\xi_h) = \overline{u_h}(\overline{\psi_{x_h,\xi_h,h}})= \overline{u_h(f_h)}.
 $$
 Since $\psi_{x_h,\xi_h,h}$ is microlocall uniformly exponentially small in the complement of $B_\epsilon(x,\xi)$
 this shows that $T_h u_h$ is uniformly exponentially small on $B_{\frac{\epsilon}{2}}(x,\xi)$.
 \end{proof}

This shows that the microsupport of a polynomially bounded family of distributions is a well defined subset of the cotangent bundle, and in fact we can use it to define the microsupport of a polynomially bounded family $(u_h) \in \mathcal{D}'_h(M)$ on a general real analytic manifold $(M,g)$. {We note that there are other coordinate independent approaches to the microsupport, in particular the theory by Sj\"ostrand that defines coherent states on manifolds by their analytic properties (\cite{MR699623}).}

\begin{definition}
 Let $(u_h)$ be a polynomially bounded family of distributions $(u_h) \in \mathcal{D}'_h(M)$. Then $\muS(u_h)$ is the complement in $T^*M$ of the set of points $(x,\xi)$ such that there exists
 an open neighborood $\calO \subset T^*M$ of $(x,\xi)$ and a function $\chi \in C^\infty_0(M)$ that equals one near $\pi(\calO) \subset M$ and the following holds.
 For all families $(f_h) \in C^\infty_{0,h}(M)$, uniformly microsupported in $\calO$ we have
 $\|u_h(\chi f_h)\| = O(e^{-\delta h^{-1}})$ for some $\delta>0$ as $h \searrow 0$.
\end{definition}

The proposition above simply states that this can also be defined in local coordinates using the notion of microsupport for Schwartz distributions by multiplying with an appropriate cut-off function.

\subsection{Analytic wavefront sets}

\begin{definition}
 The {\sl analytic wavefront set} $\WFa(u)$ of a distribution $u \in \mathcal{D}'(M)$ is the set of all $(x,\xi) \in T^*M \setminus 0$
 such that $(x,\xi) \in  \muS(u)$. Here $u$ is considered as a constant family.
\end{definition}

In fact, for a constant family of distribution the microsupport equals
$$
 \muS(u) = \mathrm{supp}(u) \times \{0\} \cup \WFa(u).
$$

The analytic wavefront set is a conic set, just as the wavefront set (\cite{MR1872698}*{Prop. 3.2.5}). The above definition is part of a zoo of definitions that we review and link in Appendix \ref{appa}. Roughly speaking the analytic wavefront set captures in which direction a distribution fails to be a real analytic function. In particular, if the analytic wavefront set $\WFa(u)$ contains no points over an open neighborhood of a point $x_0$, then $u$ is a real analytic function near $x_0$.
The analytic wavefront set transforms like a subset of the cotangent bundle under analytic changes of coordinates and is therefore well suited as a notion for distributions on real analytic manifolds.
We now refer to Appendix \ref{appa} and \ref{comprules} for further properties of the analytic wavefront set.

The following lemma can be interpreted, intuitively, as showing that test functions with compactly supported Fourier transforms carry asymptotically no energy and therefore are microlocally exponentially small away from zero. This can be made quite precise as follows.

\begin{proposition} \label{alemma}
 If $f \in \mathcal{S}(\R^d)$ is independent of $h$ and $\hat f$ has compact support then $f=f_h$ is uniformly microsupported in the set $\R^d \times \{0\} \subset \R^d \times \R^d$.
\end{proposition}
\begin{proof}
 Let $c>0$.
 The support properties of $\hat f$ and $\mathcal{F}_h f = h^{-\frac{d}{2}}  \hat f(h^{-1} \xi)$ imply that there exist a $c>0$ such that $\mathcal{F}_h f(\xi_0)=0$ when $| \xi_0 | \geq c h$.
 We have $T_h f(x,\xi)= e^{\frac{\rmi}{h} x \xi} (T_h \mathcal{F}_h f)(\xi,-x)$, and therefore 
 \begin{align*}
   T_h f(x,\xi) &= \alpha_h \int_{\R^d}  e^{-\frac{(\xi-\xi_0)^2}{2h}} e^{-\frac{\rmi}{h} x \xi_0} \hat f(h^{-1} \xi_0) \der \xi_0 \\ &=
 \alpha_h h^{-\frac{d}{2}} \int_{| \xi_0 | \leq c \cdot h}  e^{-\frac{(\xi-\xi_0)^2}{2h}} e^{-\frac{\rmi}{h} x \xi_0} \hat f(h^{-1} \xi_0) \der \xi_0.
 \end{align*}
  For multi-indices $\alpha, \beta \in \mathbb{N}_0^d$ integration by parts actually gives
 \begin{align*}
   \xi^\beta x^\alpha T_h f(x,\xi) &=  \xi^\beta  \alpha_h \int_{\R^d}  e^{-\frac{(\xi-\xi_0)^2}{2h}} \left( (\rmi h \partial_{\xi_0})^\alpha e^{-\frac{\rmi}{h} x \xi_0} \right) \hat f(h^{-1} \xi_0) \der \xi_0 \\ &=  \xi^\beta
  \alpha_h h^{-\frac{d}{2}} \int_{| \xi_0 | \leq c \cdot h}  e^{-\frac{\rmi}{h} x \xi_0}  (-\rmi h \partial_{\xi_0})^\alpha \left( e^{-\frac{(\xi-\xi_0)^2}{2h}} \hat f(h^{-1} \xi_0) \right) \der \xi_0
 \end{align*}

 Using $|\xi - \xi_0 | \geq | | \xi | - | \xi_0 | | $ this implies for $| \xi | > \tilde \delta + c \cdot h> c\cdot h$ that  
  $$
  \| \xi^\beta x^\alpha T_h f(x,\xi) \| \leq C_{\alpha, \beta}\alpha_h h^{-\frac{3d}{2}} (1+| \xi |)^{|\beta|+|\alpha|} e^{-\frac{(| \xi |-c \cdot h)^2}{2h}} 
  \sum_{\gamma \leq \alpha}\| \partial^{\gamma} \hat f \|_{L^1(\R^d)}.
 $$
 We have applied the product rule to the expression
 $$
  (-\rmi h \partial_{\xi_0})^\alpha \left( e^{-\frac{(\xi-\xi_0)^2}{2h}} \hat f(h^{-1} \xi_0) \right)
 $$
 and accumulated extra appearing factors into $C_{\alpha, \beta}$.
 This implies the claimed exponential decay uniformly when $|\xi|\geq\epsilon>0$.
 \end{proof}
 
 \subsection{The algebra of zero energy functions.} For a general analytic manifold we 
define the space  $C^\infty_{0,h,O}(M)$ of families
of compactly supported test functions that are uniformly exponentially small away from any neighborhood of the zero section in $T^*M$.
Intuitively, these test functions  asymptotically do not carry energy. By Prop. \ref{multmic}
the space $C^\infty_{0,h,O}(M)$ is a subalgebra of $C^\infty_{0,h}(M)$. 
If $M$ is compact the functions in $C^\infty_{0,h,O}(M)$ that are independent of $h$ are precisely the real analytic functions. In case $M$ is connected and not compact the zero function is the only function in  $C^\infty_{0,h,O}(M)$ that is independent of $h$.
There are however a lot of families of functions in $C^\infty_{0,h,O}(M)$ that depend on $h$. 
Given a compactly supported smooth function $\chi \in C^\infty_0(\R^d)$ that is one in a neighborhood of a point $y \in \R^d$ the family
$$
 \chi(x) e^{-\frac{(x-y)^2}{2h}}
$$
is an example of a family in $C^\infty_{0,h,O}(\R^d)$. Using analytic coordinate charts one can use these cut-off Gaussians to construct such families on general analytic manifolds $M$. More generally we have the following.
\begin{lemma}[Existence of bump functions] \label{bumplemma}
 Let $M$ be a real analytic manifold, and let $K \subset M$ be a compact subset. Then there exists a family $(\chi_h) \in C^\infty_{0,h,O}(M)$ such that $\chi_h = 1$ in an open neighborhood of $K$.
\end{lemma}
\begin{proof}
 We first show this for $M = \R^d$. Since every compact subset is contained in a compact ball it is sufficient to construct such a family for a closed ball centered at $0$. To show this we construct a family $(\chi_h) \in C^\infty_{0,h,O}(\R^d)$
 such that $\chi_h =1$ on a ball $B_R(0)$ of radius $R>0$ centered at zero.
 We choose a compactly supported bump function $b \in C^\infty_0(\R^d)$ that equals one on $B_{R+\delta}(0)$. We choose another bump function $\tilde \chi \in C^\infty_0(\R^d)$ such that $\tilde \chi$ is one near $x=0$ and supported in the ball $B_\delta(0)$. 
 {
 Then the family $\tilde \chi_h(x) = \tilde \chi(x)(2 \pi h)^{-\frac{d}{2}} e^{-\frac{x^2}{2h}}$ is supported in $B_\delta(0)$, and is an element of $C^\infty_{0,h,O}(\R^d)$. Note that
 $$
  c_h = \int_{\R^d}  \chi_h(x) \der x = 1 + r_h,
 $$
 with an exponentially small remainder term $r_h$. Hence $c_h^{-1} - 1$ is exponentially small.}
Now define $\chi_h$ as the convolution
 $$
  \chi_h  = c_h^{-1} \tilde \chi_h *b.
 $$
 This function will be one on $B_R(0)$. To see that $(\chi_h) \in C^\infty_{0,h,O}(\R^d)$ observe that the FBI transform of the convolution is $T_h(\chi_h)(x,\xi) = \int_{\R^d} (T_h \tilde \chi_h)(x-y,\xi) b(y) \der y$.
 
 To construct such a function on a general manifold we choose an analytic proper embedding $M \to \R^d$. Such an embedding always exists (see \cite{MR98847}). The statement then follows from the fact that restrictions of functions in $C^\infty_{0,h,O}(\R^d)$ to $M$ are automatically in $C^\infty_{0,h,O}(M)$. This is an immediate consequence of Prop. \ref{cochart}.
\end{proof}

We also define the space $C^\infty_{h,O}(\R^d)$ as the space of functions $f_h$ in $C^\infty_h(\R^d)$ with the property that for any compact subset $K \subset M$ and any compactly supported smooth function $\chi \in C^\infty_0(M)$ that equals one in an open neighborhood of $K$ we have for any $\epsilon>0$ the estimate 
$$
  \| (1+ |\xi|)^N T_h(\chi f_h)(x,\xi) \| \leq C_N e^{-\delta h^{-1}} 
$$
uniformly on $K \times (\R^d \setminus B_\epsilon)$. Again, using local charts one can define the algebra $C^\infty_{h,O}(M)$.
For $f_h \in C^\infty_{h,O}(M)$ and $u_h \in C^\infty_{0,h,O}(M)$ we then have $f_h u_h \in C^\infty_{0,h,O}(M)$.

Recall that a time function is a function whose gradient
is everywhere timelike, and whose level surfaces are Cauchy hypersurfaces.

\begin{lemma} \label{ncbumplemma}
 Let $(M,g)$ be a globally hyperbolic analytic spacetime.
 Suppose that $\iota: M \to \R^d, x \mapsto (\iota_1(x),\ldots,\iota_d(x))$ is a proper analytic embedding such that $t(x) = \iota_1(x)$ is a proper time function.  Then, for any $\delta>0$ there exists a family of functions $(\rho_h) \in C^\infty_{h,O}(M)$ such that
 $\rho_h(x)=1$ for all $x \in M$ with $\iota_1(x) \geq \delta$ and $\rho_h(x) =0$ for all $x \in M$ with $\iota_1(x) \leq -\delta$.
 \end{lemma}
 \begin{proof}
We construct such a function on $\R^d$ with the time function $(x_1,\ldots,x_d) \mapsto x_1$, as in the more general case this function can simply be restricted to $M$. The estimate
$$
  \| (1+ |\xi|)^N T_h(u)(x,\xi) \| \leq C_N e^{-\delta h^{-1}} 
$$
is easily checked to hold for the constant function $u=1$. 
The FBI-transform of a function of product type of the form $f(x_1,\ldots,x_d) = f_1(x_1) \cdots   f_d(x_d)$ is the product of the individual FBI-transforms.
It is sufficient to construct the family of functions $(\rho_h)$ in the case $d=1$. In the higher dimensional case we simply take the pull-back under the projection map $(x_1,\ldots,x_d) \to x_1$. This results in a product-type function of the above form, with all the functions except the first being the constant function one. Therefore, assume without loss of generality that $d=1$.

Next we observe that the above estimate is local in the sense that if a function $u$ vanishes near a point $x$, then the FBI-transform satisfies the above estimate near that point. This also implies that it is sufficient to check the estimate for an open cover.
We now use a bump function $\tilde \rho_h \in C^\infty_{0,h,O}(\R)$ as constructed in Lemma \ref{bumplemma} that equals one on the interval $[0,1]$ and has support in $[-1,2]$. Now simply define
$$
 \rho_h(x) = \begin{cases} \tilde \rho_h(x) & x \leq 1\\ 1 & x >1  \end{cases}.
$$
By the above this function satisfies locally the above estimate, and therefore is in $C^\infty_{h,O}(\R)$. This function satisfies the required properties with $\delta=1$. A simple rescaling argument shows that such a function exists for arbitrary $\delta>0$.
 \end{proof}

\section{Analytic states in Quantum Theory} \label{seclytic}

We now look at time-evolution in quantum physics from a microlocal point of view. It is instructive to look at the various wavefront sets that are naturally associated to this setting. First consider a unit vector (state) $\psi \in \mathcal{H}$ in a Hilbert space.
Let $H$ be a self-adjoint operator, the Hamiltonian that we think of as the generator of time-translations. A typical stability assumption of quantum physics is that $H$ should be semi-bounded below, i.e.  the spectrum of $H$ should be contained in the set $[-C,\infty)$ for some constant $C$. With
$$
 U(t) \psi = e^{-\rmi t H} \psi
$$
we then have $\WFa(U(t) \psi ) \subseteq \R \times (-\infty,0]$. 
To see this note that by the spectral theorem for self-adjoint operators we can pass via a unitary transformation to a Hilbert space on which $H$ is a multiplication operator by a function $f$ on some measure space $(X,\mu)$ that acts by multiplication on the Hilbert space $L^2(X,\mu)$. Then, $U(t) \psi$ is unitarily equivalent to the Hilbert space valued function
$$
 e^{-\rmi t f(x)} \psi(x),
$$
and we can compute the FBI transform in the $t$-variable at the point $t_0$
\begin{align*}
 &\TT_h(e^{-\rmi t f(x)} \psi(x))(t_0,\eta) 
 \\&= \alpha_h \int_{\R} e^{-\frac{1}{2h}(t-t_0)^2} e^{-\rmi t f(x)} \psi(x) e^{-\frac{\rmi}{h} (t-t_0) \eta} \der t  \\&=
 \alpha_h \sqrt{2 \pi h} e^{-\frac{1}{2h} (h f(x) + \eta) (h f(x) + \eta +2 \rmi  t_0) +\frac{\rmi}{h} t_0 \eta} \psi(x),
\end{align*}
and therefore
\begin{align*}
 \|\TT_{h}(e^{-\rmi t f(x)} \psi(x))(t_0,\eta) \| = 
 \alpha_h \sqrt{2 \pi h}  \| e^{-\frac{1}{2 h} (h f(x) + \eta)^2} \psi(x)\|.
\end{align*}
Since $f(x)>-C$ this is exponentially fast decaying independent of $t_0$ as $h \to 0$ as soon as $\eta>0$. Consequently the analytic wavefront set is contained in $\R \times (-\infty,0]$. Note that the analytic wavefront set is however empty if and only if the vector $\Psi$ is an analytic vector with respect to the time evolution.
It is a general result that there is always a dense set of analytic vectors for any self-adjoint operator, and more generally for any Lie group represented in a strongly continuous fashion on the Hilbert space. A generic vector is typically not analytic.

 It is instructive to investigate the model of Schr\"odinger quantum mechanics, which we think of as a quantum field in one space-time dimension. 
 In the following let $H$ be the Friedrichs extension of the operator
 $$
   - \Delta + V(x),
 $$
 on  $\mathcal{S(\R^d)}$, where $V$ is a polynomially bounded smooth potential satisfying $V(x) \geq a x^2$ for some $a>0$.  It is easy to check, using the explicit description of the domain of the Friedrichs extension, that 
 $$
  \mathrm{dom}(H) = \{ \phi \in L^2(\R) \mid \phi \in H^2(\R), \; V \phi \in L^2(\R) \},
 $$
 and it follows from the variational principle that there is a spectral gap $\mathrm{spec}(H) \subset [\sigma_0,\infty)$ for some $\sigma_0>0$.
Recall that a vector $\psi$ is analytic for $H^\alpha$ if and only if
$$
 \sum_{n=0}^\infty \frac{\| H^{n \alpha} \psi \|}{n!} t^n
$$
 converges for some $t>0$.
 {
 Since the spectral measure $\der E_\lambda$ is supported on $[\sigma_0,\infty)$ we have for any $ \alpha  \in (0,1]$ the estimate
$$
  \| H^{n \alpha} \psi \|^2 = \int_{\sigma_0}^\infty \lambda^{2n \alpha} \langle \psi, \der E_\lambda \psi \rangle \leq 
  \sigma_0^{2n(\alpha-1)} \int_{\sigma_0}^\infty \lambda^{2n} \langle \psi, \der E_\lambda \psi \rangle =\sigma_0^{2n(\alpha-1)}  \| H^{n} \psi \|^2 .
$$
This shows that, for any $\alpha \in (0,1]$, the set of analytic vectors for the operator $H^\alpha$ is contained in the set of analytic vectors for the operator $H$. We will now refer to this set simply as the set of analytic vectors without reference to $\alpha \in (0,1]$. We note however that there is a dense set of vectors which is analytic for all positive powers of the operator $H$.}
Any analytic vector $\psi$ is in the domain of the operator $e^{s H^\frac{1}{2}}$ for $s \in (-\delta,\delta)$ and sufficiently small  $\delta>0$.
This implies that $\psi_s = e^{s H^\frac{1}{2}} \psi$ satisfies the elliptic equation $(-\partial_s^2 + H) \psi_s =0$ and therefore $\psi_s(x)$ is smooth in $s$ and $x$. In case $V$ is real analytic we can conclude, by the same argument, that $\psi$ a real analytic $L^2$-function on $\R$.\\
We now consider the time evolution $U(t) = e^{-\rmi t H^\alpha}$. The case $\alpha=1$ corresponds to non-relativistic Schr\"odinger mechanics, whereas the case $\alpha=\frac{1}{2}$ is a relativistic counterpart. Since $\alpha$ will be fixed throughout we will suppress it in the notations. 

The domain of $x$, regarded as a multiplication operator in $L^2(\R)$, contains the domain of $H$. Hence, the operator $A = H^{-1} x H^{-1}$ is a bounded self-adjoint operator.
Given $f \in \mathcal{S}(\R)$ we define $A_f = \int_\R f(s) U(-s) H^{-1} x H^{-1} U(s) \der s$ as a Bochner integral.
Since $U(s)$ commutes with $H$ the unbounded operator defined by $x_f = H  A_f H$ defines the formal
expression $x_f= \int_\R f(s) U(-s) x  U(s) \der s$ avoiding the discussion of integrals of unbounded operators.

\begin{proposition} \label{invharm}
If $f \in \mathcal{S}(\R)$ is a Schwartz function, then $x_f$ leaves the domain of smoothness $\mathrm{dom}(H^\infty)$ of $H$ invariant.
For any $\psi \in \mathrm{dom}(H^\infty)$ and $m \in \mathbb{N}$ the map 
$$
\mathcal{S}(\R) \to \mathrm{dom}(H^m), f \mapsto x_f \psi
$$ 
is continuous and therefore defines a tempered vector-valued distribution.
If in addition $f$ has a compactly supported Fourier transform $\hat f \in C^\infty_0(\R)$, then $x_f$ leaves the set of analytic vectors invariant. 
 \end{proposition}
\begin{proof}
We first note that the group $U_\alpha(t)$ and the operators $H^\beta$ leave the set of analytic vectors and the domain of smoothness invariant for all $\alpha,\beta>0$.
 Therefore, we only need to show that these sets are invariant under the action of $A_f$. We prove first that the set of analytic vectors is invariant under $A_f$ if $f$ is the Fourier transform of a compactly supported smooth function.
 Let $\psi$ be an analytic vector. We need to show that $U(t)  A_f \psi$ is real analytic in $t$ near zero. We have
\begin{align} \label{formsmooth}
 U(t)  A_f \psi = \int_{\R} f(s) U(-s+t) A U(s) \der s\, \psi =  \left( \int_{\R} f(s+t) U(-s) A U(s) \der s \right) U(t) \psi.
\end{align}
By assumption the function $U(t) \psi$ is real analytic in $t$. It is therefore sufficient to establish that 
$B(t)= \int_{\R} f(s+t) U(-s) A U(s) \der s$
is an analytic function of $t$ with values in the bounded operators. Since $\| U(-s) A U(s) \| = \| A \|$ is a bounded function of $s$, the integrand is Bochner-integrable. Analyticity of $B(z)$ now follows from the fact that $f$ is entire and $f(x+z)$ is a complex analytic family of $L^1$-functions.\\
The invariance of the domain of smoothness is shown in a completely analogous way, replacing analyticity by differentiability. In case of general $f \in \mathcal{S}(\R)$ the function $B(t)$ is infinitely differentiable. That $x_f \psi$ is a tempered $\mathrm{dom}(H^m)$-valued distribution follows from the fact that $A_f$ is a tempered distribution taking values in the bounded operators from $\mathrm{dom}(H^m)$ to $\mathrm{dom}(H^m)$. One infers this directly from 
$$
 \| H^m A_f \psi \| =  \| \frac{\der^m}{\der t^m}|_{t=0} U(t) A_f \psi \| \leq C_m \left( \| f \|_{W^{m,1}(\R)} \cdot \| (H+1)^m \psi \|_{L^2(\R)}\right),
$$
for some constant $C_m$, by \eqref{formsmooth} and the product rule.
\end{proof}

The operator $x$ corresponds to the physical measurement of position and $x(t) = U(-t) x U(t)$ is the time-dependent family of operators in the Heisenberg picture. 
For a test function $f \in \mathcal{S}(\R)$ then $x_f$ plays the role of the smeared out field operator.
Given $\psi \in L^2(\R)$ we can think of  $x(t) \psi$ as a distribution with respect to the $t$-variable on $\R$ and analyse its analytic wavefront set. We claim that if $\psi$ is an analytic vector for the group $U(t)$ then $\WFa(x(t) \psi) = \R \times [0,\infty)$. To see this note that $U(t) \psi$ is an analytic function in the variable $t$ taking values in the domain of $H$, and hence in the domain of $x$. Thus, $U(-t) x U(t) \psi$ is the boundary value of a holomorphic function near the real axis in the upper half space.
This implies, by Prop. \ref{A3}, that $\WFa(x(t) \psi) \subseteq \R \times [0,\infty)$. 
If the vector $\psi$ is not analytic $x(t) \psi$ generically has wavefront set equal to $\R \times \R \setminus \{0\}$.
In this example we have a dense set of states $\psi$ with the property that $\WFa(U(t) \psi)$ is contained in  $\R \times [0,\infty)$.

A stronger statement is obtained by considering products of these operators and their analytic wavefront sets.
Given a vector in the domain of smoothness $\mathrm{dom}(H^\infty) = \cap_{k} \mathrm{dom}(H^k)$ the formal expression
$$
 x(t_1) \cdots x(t_m) \psi
$$
can be viewed as a Hilbert space valued distribution on $\R^m$. A priori, this expression does not make sense pointwise because $\psi$ may not be in the domain of $x(t_1) \cdots x(t_m)$. By Prop. \ref{invharm} the smeared out operators $x_f$ leave the domain of smoothness invariant, and therefore the expression is well defined as a tempered distribution.
Assuming that $\psi$ is an analytic vector the analytic wavefront of $x(t_1) \cdots x(t_m) \psi$ can only contain vectors of the form
$$
 (t_1,\xi_1,t_2,\xi_2,\ldots,t_m,\xi_m)
$$
where the rightmost non-zero number $\xi_j$ is non-negative. 
This can be inferred from the following proposition.

\begin{proposition}
 Assume that $\psi$ is an analytic vector. Then the analytic wavefront set of the distribution $x(t_1) \cdots x(t_m) \psi$
 is contained in the set $$
 \{(t_1,\xi_1,\ldots,t_m,\xi_m) \mid \xi_1+\ldots+\xi_m \geq 0,  \xi_2+\ldots+\xi_m \geq 0, \ldots, \xi_m \geq 0\}.
$$
\end{proposition} 
\begin{proof}
It is convenient to change coordinates to 
$$
 z_1=t_1, z_2 = t_2-t_1,\ldots, z_m = t_{m}-t_{m-1}.
$$ 
This is a linear change and therefore induces a continuous map on Schwartz space.
In this coordinate system we have in the sense of distributions
$$
 x(t_1) \cdots x(t_m) \psi = U(-z_1) x U(-z_2) x U(-z_3) x \ldots U(-z_m) x U(z_1 + \ldots z_m) \psi.
$$
By assumption the function $U(z_1 + \ldots +z_m) \psi$ is an analytic functions in the variables $z_1,\ldots,z_m$ taking values in the domain of  $H$. We now consider the wavefront set of the distribution $U(-z_1) x U(-z_2) x U(-z_3) x \ldots U(-z_m) x$ that we regard as a tempered distribution taking values in the bounded operators from $\mathrm{dom}(H)$ to $\dom(H)$. 
To do this consider $H U(-z_1) x U(-z_2) x U(-z_3) x \ldots U(-z_m) x H^{-1}$ as a distribution with values in the bounded operators on $L^2(\R)$. 
This distribution can be written as the distributional derivative as $\partial_1 \cdots \partial_m K(z_1,\ldots,x_m)$ of the distribution
$$
K(z_1,\ldots,z_m) = U(-z_1) B U(-z_2) B U(-z_3) B \ldots U(-z_m) B,
$$
where $B$ is the bounded operator $x H^{-1}$.
Since $U(-z)$ is a bounded holomorphic function of $z$ taking values in the bounded operators in the region $\Im(z)>0$ 
we have that $K$ is the distributional boundary value of a function, holomorphic in small region $\mathcal{V}$ around the real axis intersected with 
$$
 \{(z_1,\ldots,z_m) \mid \Im(z_1),\ldots,\Im(z_m) >0 \}.
$$
Hence, by Prop. \ref{A3}, the analytic wavefront set is contained in
$$
 \{(t_1,\xi_1,\ldots,t_m,\xi_m) \mid \xi_1,\ldots,\xi_m \geq 0\}.
$$
Pull back with respect to the above change of coordinates shows that the analytic wavefront set is contained in the set
$$
 \{(t_1,\xi_1,\ldots,t_m,\xi_m) \mid \xi_1+\ldots+\xi_m \geq 0,  \xi_2+\ldots+\xi_m \geq 0, \ldots, \xi_m \geq 0\}.
$$
The same must be true for all distributional derivatives of $K$.
\end{proof}

\section{Quantum fields on spacetimes} \label{qftsec}

In the following we assume that $(M,g)$ is a connected $n$-dimensional spacetime (time-oriented, oriented Lorentzian manifold). We assume further that $M$ has a real analytic atlas with respect to which the metric $g$ is real analytic. The purpose of the metric $g$ is two-fold. It provides an analytic volume form which is needed to identify functions with distributions. Secondly, it
provides a causal structure in the form of a bundle of light cones. In principle it is not strictly necessary to derive these structures from a metric, but we assume so here for simplicity of the presentation.

We briefly explain the notations, assuming the signature convention is $(+,-,\ldots,-)$.
First, for a point $x \in M$ we denote by $V_x \subseteq T^*_x M$ the closed light-cone in cotangent space with respect to the metric $g$, i.e. the set of covectors $\xi \in T^*_x M$ with $g^{-1}(\xi,\xi) \geq 0$.
Then $V_x \setminus \{0\}$ is the disjoint union of the future/past light cone $V_x^\pm$. A covector $\xi \in T^*_x M$ is called causal if it is in $V_x \setminus \{0\}$. A causal covector is called future directed if it is contained in $V_x^+$.
We write $V^\pm$ for the corresponding bundles, i.e. $V^\pm = \sqcup_x V_x^\pm$. Therefore, $(x,\xi) \in V^+$ will mean that $\xi$ is a future directed causal covector.
Recall that $T^*(M \times \ldots \times M)$ is canonically isomorphic to $T^*M^m$ and we will write typical elements as $(x_1,\xi_1,\ldots,x_m,\xi_m)$.

\subsection{Quantum field theory}

A quantum field theory on $(M,g)$ will be defined as an operator-valued distribution $\Phi$, which we will call the field.
To be more precise, let $\calH$ be a Hilbert space with a dense set $D$. Then $\Phi$ is a map
$\Phi: C^\infty_0(M) \to \mathrm{End}(D)$ such that $f \mapsto \Phi(f) v$ is continuous for every $v \in D$. The operator $\Phi(f)$ can be unbounded. It is referred to as the smeared out field. One requires that $\Phi(f)$ is symmetric on $D$ if $f$ is real-valued. We will assume for convenience and without loss of generality that $D$ is complete with respect to the locally convex topology induced by the family of semi-norms
$$
 p_{f_1,\ldots,f_m}(\phi) = \| \Phi(f_1) \cdots  \Phi(f_m) \phi \|,
$$
where $m \in \mathbb{N}_0$ and $(f_1,\ldots,f_m)$ is an arbitrary $m$-tupel of compactly supported smooth functions.
We refer to this locally convex topology as the graph topology, as it is generated by the graph norms of all the elements of the algebra. Since the adjoints are densely defined the operators $\Phi(f_1) \cdots  \Phi(f_m)$ are closable. Therefore, one can always pass to the completion of the domain, which is then still contained in $\mathcal{H}$. 

Given a spacetime region $\calO \subseteq M$ one can form the $*$-algebra $\calA(\calO)$ generated by the elements $\Phi(f), f \in C^\infty_0(\calO)$.
In case $\Phi(f)$ is essentially self-adjoint for any real-valued $f \in C^\infty_0(\calO)$ one can then consider the weak-$*$-closure of the set of all bounded operators generated by the spectral projections of $\Phi(f)$. By von-Neumann's bi-commutant theorem this algebra can be characterized as
$$
 \calR(\calO) = \{ A \in \mathcal{L}(\calH) \mid \forall  f \in C^\infty_0(\calO), A \textrm{ commutes with }\Phi(f) \}'.
$$
Here $\mathcal{R}' = \{B \in \calL(\calH) \mid \forall A \in \mathcal{R}, A B = B A\}$ as usual is the commutant of $\mathcal{R}$.
Recall that a bounded operator $A$ commutes with a self-adjoint operator $T$ if and only if
$A T = T A$ as an inequality of unbounded operators with equality of domains. In particular,
$A$ leaves the domain of $T$ invariant.
If $T$ is essentially self-adjoint on a dense set $D \subseteq \calH$ then an operator $A \in \calL(\calH)$ commutes with $T$ therefore if and only if for all $v,w \in D$ we have $\langle T w, A v \rangle = \langle w,A T v \rangle$.

For a general symmetric unbounded operator $T$ defined on a dense set $D$ we turn this into a definition and say a bounded operator $A \in \calL(\calH)$ commutes (weakly) with $T$ if for all $v,w \in D$ we have $\langle w,  A\, T v \rangle = \langle T  w,  A\, v \rangle$. The set of operators commuting with $T$ is then a set that is invariant under the map $*$. It is called the weak-commutant of $T$. It is easy to see that if $A$ commutes with $T$ then it also commutes with its closure. It is therefore sufficient to check commutation on a subset of the domain that is dense in the graph norm, so that the closure of the operator on this subset coincides with the operator. 
It is not sufficient to check this on a dense subset of the domain. As an example consider the Laplace operator $\Delta$ on the real line and the Laplace operator $\Delta_D$ on the real line with Dirichlet boundary conditions at the point $x=0$. If we take $T = \Delta$ and $A=(-\Delta_D+1)^{-1}$, then $T$ and $A$ do not commute. However we have the relation
$\langle w,  A\, T v \rangle = \langle T  w,  A\, v \rangle$ for all $v,w$ in the dense set of smooth compactly supported functions that vanish to infinite order at the point $0$. The reason is here that  $\Delta$ and  $\Delta_D$ restrict to the same operator on this space of functions, but the self-adjoint extensions are completely different. This shows that one has to consider domain issues carefully when using this definition and associated conclusions.
For the functional analytic details we refer to  \cite{MR283580} as well as \cite{MR1199168} for a discussion in the context of quantum field theory.
 
The {\sl von-Neumann algebra (weakly) associated with the algebra} generated by
$\Phi(f), f \in C^\infty_0(\calO)$ is then defined by
$$
 \mathcal{R}(\calO) = \{ A \in \mathcal{L}(\calH) \mid \forall  (f \in C^\infty_0(\calO), v,w \in D), \langle w,  A\, \Phi(f) v \rangle = \langle \Phi(\overline{f})  w,  A\, v\rangle \}'.
$$

In this way every quantum field defines a net of von Neumann algebras, i.e. a von Neumann algebra  $\mathcal{R}(\calO)$ associated to every spacetime region $\calO$. Due to the nature of the weak commutant there are further technical conditions that ensure that this is again a local net, i.e. that it satisfies Einstein causality and the algebras of causally separated regions commute. We will however not discuss this here any further but refer to \cite{MR847127} for a detailed discussion of this in the context of Wightman fields on Minkowski space.

We now discuss a mild assumption about the quantum field. 

\subsection{Physical conditions on states} \label{microlocal}

Whereas on a curved spacetime there is no meaningful notion of momentum and energy and hence no preferred vacuum state, the notion of energy-momentum should still exist in an asymptotic sense as a scaling limit. The notion of uniform microsupport for test functions is extremely well suited to capture this. On physical grounds one expects from a reasonably passive state $\Omega$ to not allow for non-physical energy transfer. Assume that $(f_h)$ localizes in phase space to a point $(x,\xi)$ as $t \searrow 0$. For the moment we also allow families $(\phi_h)$ of states in $D$ describing an asymptotic physical situation.
We think of $\phi_h$ as a physical configuration that may become singular as $h \searrow 0$.
In QFT we then expect the following behaviour depending on where in $T^*M$ the point $(x,\xi)$ is localized.
If $(x,\xi) \in V^-$ we expect the state $\Phi(f_h) \Omega$ to have added energy-momentum $-h^{-1} \xi$ near $x$ in an asymptotic sense.
If $(x,\xi) \in V^+$ we expect the operator $\Phi(f_h)$ to erase the energy-momentum $h^{-1}\xi$ near $x$ from the state $\phi_h$.
A state $\Omega \in D$ not asymptotically carrying energy should not allow for asymptotic energy extraction. Hence,
$\Phi(f_h) \Omega$ should be exponentially small if $(x,\xi) \notin V^-$.

The mathematically precise statement depends on the space of test functions employed and we would like to consider two versions.

\subsection{Fields defined on compactly supported smooth test functions}

Recall that   $C^\infty_{0,h,O}(M)$ is the space of families
of test functions that are uniformly exponentially small away from any neighborhood of the zero section in $T^*M$.
Given $(q_{1,h},\ldots,q_{n,h}) \in C^\infty_{0,h,O}(M)$  we expect intuitvely that the state 
$\Phi(q_{1,h})\cdots \Phi(q_{n,h}) \Omega$ contains no asymptotically extractable energy. 
We therefore should have that
$$
 \Phi(f_h)\Phi(q_{1,h})\cdots \Phi(q_{n,h}) \Omega
$$
is exponentially small if $f_h$ microlocalizes uniformly at a point  $(x,\xi) \notin V^-$.
This motivates the following definition.
\begin{definition}
 A vector $\Omega \in D$ is called analytic if the following holds. If
 $f_h \in C^\infty_{h}(M)$ is microlocally uniformly supported in a compact set $K \subset T^*M$ with $K \cap V^- = \emptyset$ then, 
 for all families  $(q_{1,h}),\ldots,(q_{n,h}) \in C^\infty_{0,h,O}(M)$, we have the bound
 $$
  \| \Phi(f_h) \Phi(q_{1,h})\cdots \Phi(q_{n,h}) \Omega \| \leq  C e^{-\delta h^{-1}}
 $$
 for some $C>0, \delta>0$.\\
 The subspace of analytic vectors will be denoted by $D_a \subset D$.
 \end{definition}

Using Prop. \ref{secondchar} the above condition can be completely paraphrased in terms of analytic wavefront sets.
 
 \begin{proposition} \label{altcharanwf}
 A vector $\Omega \in D$ is analytic if and only if for all $n \in \mathbb{N}$ the Hilbert space valued distribution $\Phi(\cdot) \cdots \Phi(\cdot)  \Omega$ on $M^n$ defined by
 $$
  f_1 \otimes \ldots \otimes f_n \mapsto \Phi(f_1) \cdots \Phi(f_n)  \Omega
 $$
 has its analytic wavefront set contained in the set of non-zero covectors $(x_1,\xi_1,\ldots x_n,\xi_n)$ satisfying for all $1 \leq k \leq n$ that
 $$
 (\textrm{if } \xi_j=0 \textrm{ for all } j>k ) \textrm{ then } \xi_k \in V^+,   
 $$
 In other words the first non-zero covector from the right must be future directed and causal.
  \end{proposition} 
\begin{proof}
 That the wavefront set condition for $k=1$ implies analyticity is a direct consequence of Prop. \ref{secondchar}. 
  We therefore only need to show that analyticity implies the analytic wavefront set condition for any $1 < k \leq n$.
   To show that the wavefront set condition is satisfied it is sufficient, again by Prop. \ref{secondchar}, to show that
  for families $f_{1,h}, \ldots f_{n,h} \in C^\infty_{0,h}(M)$ with the properties
  \begin{itemize}
   \item $(f_{k,h})$ is uniformly microsupported in a compact set of positive distance to $V^-$,
   \item $(f_{j,k}) \in C^\infty_{0,h,O}(M)$ for all $j >k$, 
  \end{itemize}
  we have that the family of vectors $ \Phi(f_{1,h})\cdots \Phi(f_{k-1,h}) \Phi(f_{k,h}) \Phi(f_{k+1,h}) \cdots \Phi(f_{n,h}) \Omega$ is exponentially small in the norm as $h \searrow 0$.
 This is equivalent to 
 \begin{align*}
    &\langle \Phi(f_{1,h})\cdots \Phi(f_{k-1,h}) \Phi(f_{k,h})\cdots \Phi(f_{n,h}) \Omega, \Phi(f_{1,h})\cdots \Phi(f_{k-1,h}) \Phi(f_{k,h}) \cdots \Phi(f_{n,h}) \Omega \rangle \\&=
     \langle  \Phi(\overline{f_{k-1,h}})\cdots \Phi(\overline{f_{1,h}}) \Phi(f_{1,h})\cdots \Phi(f_{k-1,h}) \Phi(f_{k,h}) \cdots \Phi(f_{n,h}) \Omega,  \Phi(f_{k,h})  \cdots \Phi(f_{n,h}) \Omega \rangle
 \end{align*}
 being exponentially small. Since the family of test functions 
 $$
  \overline{f_{k-1,h}} \otimes \ldots \otimes \overline{f_{1,h}} \otimes f_{1,h} \otimes \ldots \otimes f_{k-1,h} \otimes f_{k,h}  \otimes \ldots \otimes f_{n,h}
 $$
 is polynomially bounded the Cauchy-Schwarz inequality shows that exponential smallness is implied by exponential smallness of the family
 $$
  \Phi(f_{k,h}) \cdots \Phi(f_{n,h}) \Omega.
 $$
 By Prop. \ref{secondchar} exponential smallness of this vector as $h \searrow 0$ is a consequence of analyticity.
\end{proof}  
   
For a realistic quantum field we expect the set $D_a$ to be dense in the Hilbert space and in the domain of the field in the following sense. For every $v \in D$ there exists a sequence $v_n \in D_a$ such that for all $ f_1,\ldots,f_m \in C^\infty_0(M)$ we have in the Hilbert space norm
 \begin{align*}
  v_n \to v, \textrm{ and } \Phi(f_1)\cdots \Phi(f_m)  v_n \to \Phi(f_1)\cdots \Phi(f_m) v.
 \end{align*}
  This means that $D_a$ is a dense subset in each of the domains of  $\Phi(f_1)\cdots \Phi(f_m)$ with respect to the graph norm.

\begin{rem}
We note that whereas we assume that $D$ is invariant under $\phi(f)$ this cannot be assumed for $D_a$, as this is not compatible with Einstein causality. To illustrate this we restrict this discussion to bosonic fields, for which the fields commute at spacelike separation.
Assume $\phi$ was a vector in $D_a$ such that $\Phi(f) \phi \in D_a$ for all $f \in C^\infty_0(M)$. Then the distribution $[\Phi(\cdot), \Phi(f)] \phi$ vanishes in the causal complement of the support of $f$ and has its analytic wavefront set in $V^+$. Since $V^+$ is one-sided this distribution has the unique continuation property and must therefore vanish (c.f. Prop. \ref{ucp}). It follows that $[\Phi(\cdot), \Phi(f)] \phi$ vanishes.  Since this is true for all $f$ with sufficiently small support it follows for all compactly supported $f$. Hence, $[\Phi(f_1), \Phi(f_2)] \phi =0$ for all test functions $f_1, f_2 \in C^\infty_0(M)$. The existence of the dense and invariant set of analytic vectors therefore implies that the field algebra is commutative.
\end{rem}

To keep the notations short we write $\calA$ for the algebra generated by $\Phi(f), f \in C^\infty_0(M)$. Given an open set $\calO \subseteq M$ we write $\calA(\calO)$ for the algebra generated by $\Phi(f), f \in C^\infty_0(\calO)$. 
Recall that a vector $\phi \in D$ is called {\sl cyclic} for an algebra of operators $\mathcal{B}$ on $D$  if the set $\mathcal{B} \phi$ is dense in $\calH$. 
In Minkowski theories, one usually assumes that the vacuum $\Omega$ is a cyclic vector for the field algebra, and in fact that $\calA \Omega$ is dense in $D$ with respect to the graph topology.
It is also natural to assume that there are many vacuum-like states $\Omega \in D_a$ in the sense that there is a dense set of cyclic analytic vectors.

The existence of an analytic vector $\Omega \in D_a$ with $\calA \Omega$ being dense in $D$ with respect to the graph topology is one of the weaker conditions one can make, and it readily implies two important properties of the quantum field: the Reeh-Schlieder property and the timelike tube property. The Reeh-Schlieder property means that the vector $\Omega$ is a cyclic vector for the local algebra $\calA(\calO)$ for any non-empty open $\calO \subset M$. The timelike tube property is that the local von Neumann algebra $\calR(\calO)$ of region $\calO$ coincides with the local von-Neumann algebra of a potentially much bigger region, $\calE_T(\calO)$, the timelike hull of $\calO$.
We will state the precise theorems in Section \ref{tlt}.

We will see below that in case the theory satisfies a certain temperedness assumption the existence of a cyclic tempered
analytic vector implies that there is a dense set of tempered analytic vectors.

\subsection{Tempered fields}
In Minkowski space a good choice of test function space is the space of Schwartz functions $\mathcal{S}(\R^d)$.
This space is particularly suited for spectral considerations as it contains the space $\mathcal{F}(C^\infty_0(\R^d))$ of functions that are localized in momentum space. Moreover, this space treats configuration and momentum space on an equal footing.

To define an analogue of Schwartz space on a general analytic manifold $(M,g)$ one needs to specify some extra structure such as a special coordinate chart near infinity.
We will choose here a more flexible path, by embedding the spacetime analytically into $\R^{d'}$.


Any real analytic manifold $(M,g)$ can be analytically embedded into $\R^{d'}$ in such a way that the embedding is proper (see \cite{MR98847}). Given such an embedding, and restricting the space of functions to $M \subset \R^{d'}$, we can define Schwartz spaces and spaces of analytic functions on a general real analytic manifold. 
In the following we fix a proper analytic embedding $\iota: M \to \R^{d'}$. We denote by $\mathcal{S}(M)$
the space $\iota^*(\mathcal{S}(\R^{d'}))$ of the restrictions of Schwartz functions. We equip $\mathcal{S}(M)$ with the natural quotient topology. The resulting topology is stronger than the $C^\infty(M)$-topology of uniform convergence of all derivatives on compact subsets. We have now constructed nuclear Frech\'et space $\mathcal{S}(M)$ in which $C^\infty_0(M)$ is dense.
Of course, also $\mathcal{S}_a(M) = \iota^*(\mathcal{F}(C^\infty_0(\R^{d'}))$ is dense in $\mathcal{S}(M)$. This space is subspace of the space of real analytic Schwartz functions that arise from restrictions of entire functions to $M$.

It is worth noting that the map $\tilde \iota: \R^d \to \R^d$ defined as $$\tilde \iota(x_1,\ldots,x_{d}) = (x_1 \exp(x_1^2),\ldots, x_{d} \exp(x_{d'}^2))$$ is also a proper analytic embedding and the pull back $\iota^*\mathcal{S}(\R^d)$ then consists of exponentially decaying real analytic functions.
This means the embedding can always be modified so that $\mathcal{S}(M)$ consists of extremely fast decaying functions. 

For Schwartz functions a uniform notion of localisation at the zero section can be defined as follows.
We say $f_h \in \mathcal{S}_{h,O}(\R^d)$ if $f_h$ microlocalizes uniformly away from any tubular neighborhood of the form $\{(x,\xi) \in \R^{d} \times \R^d \mid \|\xi \| < \epsilon\}, \epsilon>0$.
Now define $\mathcal{S}_{h,O}(M) = \iota^* \mathcal{S}_{h,O}(\R^{d'})$ and $\mathcal{S}_a(M) = \iota^*\mathcal{S}_{a}(\R^{d'})$ by restriction. This definition is well behaved under various canonical constructions.
\begin{itemize}
\item In case $d'>d$ we have for the standard embedding $\iota: \R^d \to \R^{d'}, (x_1,\ldots,x_d) \mapsto (x_1,\ldots,x_d,0,\ldots,0)$ that $\iota^* \mathcal{S}_{h,O}(\R^{d'}) = \mathcal{S}_{h,O}(\R^{d})$. 
\item If $N$ is closed (compact without boundary) then $\mathcal{S}_{h,O}(N)$ is independent of the embedding and equal to 
$C^\infty_{0,h,O}(N)$.
\end{itemize}

The notion of $(f_h) \in C^\infty_{0,h,O}(M)$ is defined with respect to local analytic coordinates, whereas the notion of $(f_h) \in \mathcal{S}_{h,O}(M)$ is defined relative to an analytic embedding. However, for families
$(f_h) \in C^\infty_{0,h}(M)$ we have  $(f_h) \in \mathcal{S}_{h,O}(M)$ if and only if
$(f_h) \in C^\infty_{0,h,O}(M)$. Hence, the notions by analytic coordinates and by embeddings coincide. In particular this also implies that the dependence of the space $\mathcal{S}_{h,O}(M)$ on the embedding disappears upon restriction to $C^\infty_{0,h,O}(M)$. The main purpose of the embedding is thus to control the microlocal properties of the functions near infinity.

\begin{example}
 The Schwarzschild-Kruskal spacetime is a four dimensional analytic spacetime which is analytic-diffeomorphic to
 $\calO \times \mathbb{S}^2$, where  $\calO$ is the region $\{(T,X) \in \R^2 \mid T^2 - X^2<1\}$ in $\R^2$.
 The equation
 $$
  T^2 - X^2 = (1-\frac{r}{2 M})e^{\frac{r}{2 M}}
 $$
 implicitly defines a function $r(T,X)$. Then the metric is given by
 $$
  g = \frac{32 M^3}{r} e^{-\frac{r}{2 M}}(\der T^2 - \der X^2) - r^2 g_{\mathbb{S}^2}.
 $$
 The above description is called the Kruskal-Szekerez coordinate system.
 We can embed this spacetime analytically in $\R^6$ as follows. We choose the canonical embedding $\tilde \rho: \mathbb{S}^2 \to \R^3$. We embed $\calO$ into $\R^3$ by the map
 $$
  \rho: \calO \to \R^3, \quad (T,X) \mapsto (\frac{T}{r},T,X).
 $$
 Then $\iota=\rho \oplus \tilde \rho$ embeds the entire spacetime analytically into $\R^6$.
 A function in $\mathcal{S}(M)$ with respect to this embedding is a function $f: M \to \C$ that can be written in the form
 $$
  f(T,X,y) = g(\frac{T}{r},T, X,y)
 $$
 where $g \in \mathcal{S}(\R^6)$. 
 One can check that the function $\iota_1 = \frac{T}{r}$ is a global time function whose level surfaces are spacelike Cauchy hypersurfaces.
 \end{example}


\begin{definition}
A vector $\Omega \in D$ is called {\sl tempered analytic} if the following holds. If
 $f_h \in \mathcal{S}_{h}(M)$ is microlocally uniformly supported in a compact set $K \subset T^*M$ with $K \cap V^- = \emptyset$ then, 
 for all families  $(q_{1,h}),\ldots,(q_{n,h}) \in  \mathcal{S}_{h,O}(M)$, we have the bound
 $$
  \| \Phi(f_h) \Phi(q_{1,h})\cdots \Phi(q_{n,h}) \Omega \| \leq  C e^{-\delta h^{-1}}
 $$
 for some $C>0, \delta>0$.
 The subspace of analytic vectors will be denoted by $D_{ta} \subset D$.
\end{definition}

It is clear that $D_{ta} \subseteq D_a \subseteq D$. The condition of being tempered analytic seems to be a stronger conditition than that of being analytic. In particular, the existence of a tempered analytic vector readily implies that there are many other tempered analytic vectors. 
The following theorem should be compared with Prop. \ref{invharm}.

\begin{theorem}
 Suppose that there is a proper embedding $\iota: M \to \R^{d'}$ such that the quantum field 
$\Phi(\cdot)$ extends as an operator-valued distribution to the test function space $\mathcal{S}(M)$. Assume that 
$\Omega \subset D_{ta}$ is tempered analytic. Then for any collection of test functions $g_1,\ldots,g_m \in \mathcal{S}_a(M)$ the vector $\Phi(g_1)\cdots \Phi(g_m) \Omega$
 is also tempered analytic.
\end{theorem}

\begin{proof}
 This follows immediately from the inclusion $\mathcal{S}_{a}(M) \subseteq \mathcal{S}_{h,O}(M)$ and Prop. \ref{alemma}.
\end{proof}

This means the set of analytic vectors is invariant under the action of fields smeared out with certain real analytic test functions. 
Since $\mathcal{S}_{a}(M)$ is dense in $\mathcal{S}(M)$ the continuity assumption implies that the sets
\begin{align*}
 \{\Phi(g_1)\cdots \Phi(g_m) \Omega \mid g_1,\ldots,g_m \in \mathcal{S}_a(M) \},
  \{\Phi(g_1)\cdots \Phi(g_m) \Omega \mid g_1,\ldots,g_m \in \mathcal{S}(M) \}
\end{align*}
have the same closure.
In particular, if $\Omega$ is cyclic, the existence of a single tempered analytic vector implies that there is a dense set of tempered analytic vectors. 
The counterexample at the end of Appendix  \ref{comprules} shows the difficulty of proving such a statement based on a restriction on the analytic wavefront set, as in Prop. \ref{altcharanwf}, without a hypothesis of temperedness.

We now discuss the relation to Minkowski theories in more detail.

\subsection{Wightman fields in Minkowski spacetime as an example}

It is instructive to understand the above assumption in the context of Wightman field theories on Minkowski space, where it is automatically satisfied. Indeed, let $(\Phi,\calH,\Omega)$ be a Wightman quantum field theory on $d$-dimensional Minkowksi spacetime.
In this case the invariant domain $D$ would be cyclically generated from the vacuum vector $\Omega$, i.e. is the graph-closure of the span of the set of vectors of the form
$$
 \Phi(f_1)\ldots \Phi(f_n) \Omega, \quad f_1,\dots, f_n \in \mathcal{S}(\R^d).
$$
This domain is invariant by definition. It is known that the set $\tilde D_a$ defined as the span of
$$
 \Phi(f_1)\ldots \Phi( f_n) \Omega, \quad f_1,\dots, f_n \in \mathcal{S}_a(\R^d).
$$
is a dense set in the sense discussed before, and we have
$$
 \WFa(\Phi(\cdot) \phi) \subseteq V^+, \quad \textrm{ for all } \phi \in \tilde D_a.
$$
The set of vectors $\tilde D_a$ can be interpreted as the set of the states with finite spacetime momentum. Of course functions that are compactly supported in Fourier (momentum) space cannot be compactly supported in spacetime. It is therefore necessary to use Schwartz functions rather than compactly supported smooth functions as test functions. In fact, a stronger statement is true.

\begin{theorem}
 Let $(\Phi(\cdot),D \subset \mathcal{H},\Omega)$ be a (tempered) Wightman quantum field in $d$-dimensional Minkowksi spacetime satisfying the spectrum condition. Then the vector $\Omega$ is a tempered analytic vector.
\end{theorem}
\begin{proof}
 We assume that the family $(f_{k,h})$ is uniformly microsupported in the zero section $\R^d \times \{ 0\}$ if $k=2,\ldots,n$, and that $(f_{1,h})$ is a family  uniformly microsupported
in a subset $\calU$ that has positive distance from the backward light-cone $V^- =\{(x,\xi) \mid g(\xi,\xi)\leq 0, \xi_0 \leq 0\}$.
For brevity we write $N = n \cdot d$ and we introduce the following sets
\begin{align*}
  K = \{((x_1,\xi_1),\ldots,(x_n,\xi_n)) \in \R^{2N} \mid \xi_1 +\ldots+ \xi_n \in V^-\},\\
  K_\epsilon = \{(x,\xi) \in \R^{2N} \mid \mathrm{dist}((x,\xi),K) \leq \epsilon\},\\
  Q = \mathrm{pr}_2(K)=\{(\xi_1,\ldots,\xi_n) \in \R^{N}  \mid \xi_1 +\ldots+ \xi_n \in V^-\},\\
  Q_\epsilon = \{\xi \in \R^N \mid \mathrm{dist}((x,\xi),Q) \leq \epsilon\}.
\end{align*}
 
Hence, the family $(f_{1,h} \otimes \cdots \otimes f_{n,h})$ is uniformly microlocally exponentially small on
$K_\epsilon$ for some $\epsilon>0$.

We need to show that 
$$
 \| \Phi(f_{1,h})\Phi(f_{2,h}) \cdots \Phi(f_{n,h}) \Omega \|
$$ is exponentially small. This is of course equivalent to
$$
 w_{2n}(\overline{f_{n,h}}, \overline{f_{n-1,h}},\dots,\overline{f_{1,h}},f_{1,h}, f_{2,h},\dots,f_{n,h})
$$
being exponentially small. We
consider the following $h$-dependent family $(u_h) \in \mathcal{S}_h'(\R^N)$ defined by
$$
 u_h(\tilde f_1,\ldots,\tilde f_n) = w_{2n}(\overline{f_{n,h}}, \overline{f_{n-1,h}},\dots,\overline{f_{1,h}},\tilde f_1,\dots,\tilde f_{n}).
$$
We need to show that $c_h = u_h(f_{1,h}, f_{2,h},\ldots, f_{n,h})$ is exponentially small.
Now consider the inverse semi-classical Fourier transform $v_h=\mathcal{F}^{-1}_h(u_h)$. 
We obtain
\begin{equation} \label{chint}
 c_h= (v_h) (\mathcal{F}_h f_{1,h}, \ldots, \mathcal{F}_h f_{n,h}).
\end{equation}
By \eqref{csrep} we have the representation
$$
 f_{k,h} = (2 \pi)^{-\frac{d}{2}}  \int_{\R^{2d}} (T_h f_{k,h})(x,\xi) \psi_{x,\xi,h} \der x \der \xi,
$$
which gives
\begin{align*}
 c_h =  \int_{\R^{2N}}  (T_h f_{1,h})(x_1,\xi_1)\cdots(T_h f_{n,h})(x_{n},\xi_{n})
 \left( v_h(k_{x,\xi,h}) \right) \der x_1 \der \xi_1 \cdots \der x_{n} \der \xi_{n},
\end{align*}
where $k_{{h,x,\xi}}$ is the family of test functions in $\mathcal{S}(\R^N)$ defined by
$$
 k_{x,\xi,h} = (2 \pi)^{-\frac{N}{2}}  \mathcal{F}_h\psi_{x_1,\xi_1,h} \otimes \mathcal{F}_h\psi_{x_1,\xi_2,h} \otimes \cdots \otimes \mathcal{F}_h\psi_{x_{n},\xi_{n},h} 
$$
and we abbreviate $(x,\xi) = (x_1,\xi_1,\ldots, x_{n},\xi_{n})$.
Since the semi-classical Fourier transform $\mathcal{F}_h \psi_{x_0,\xi_0,h}$ of a coherent state $\psi_{x_0,\xi_0,h}$ equals
$$
\mathcal{F}_h \psi_{x_0,\xi_0,h}(\eta) = (\pi h)^{-\frac{d}{4}} e^{-\rmi x_0 \eta} e^{-\frac{(\eta-\xi_0)^2}{2h}}. 
$$
the functions $k_{h,x,\xi}$ form a family of Gaussians localising at the point $\xi$ as $h \searrow 0$. 
Now recall that $v_h$ is a polynomially bounded family of tempered distributions. This means we have the bound
$(v_h)(k_{x,\xi,h}) \leq h^{-m} p(k_{x,\xi,h}), h \in (0,1]$ in terms of a Schwartz space semi-norm $p$ for some $m>0$.
This implies that $(v_h)(k_{x,\xi,h})$ is a polynomially bounded function, i.e.
$$
 (v_h)(k_{x,\xi,h}) \leq C \left( \frac{(1 + |x| + |\xi|)}{h} \right)^M,
$$
for some $C,M>0$ and all $h \in (0,1]$. We can now split the integral \eqref{chint} 
into two parts $I_{1,h}$ and $I_{2,h}$, inserting $1-\chi$ and $\chi$ in the integral, using a smooth bounded cutoff function 
$\chi \in C^\infty(\R^{2N})$ with bounded derivatives with the following properties.
\begin{itemize}
 \item $\supp \chi$ has positive distance from $K_{\frac{\epsilon}{2}}$,
 \item $\supp (1- \chi)$ has positive distance from the complement of $K_{\epsilon}$.
\end{itemize}
Since these sets have positive distance such a function always exists.
The first integral $I_{1,h}$ is exponentially small because the family 
$(T_h f_{1,h})(x_1,\xi_1)\cdots(T_h f_{n,h})(x_{n},\xi_{n})$ is uniformly exponentially small in $K_\epsilon$.
To see that the second integral $I_{2,h}$ is exponentially small as well we note that, by the spectrum condition, $v_h$ is supported in $K$,
and therefore we can replace the test functions $k_{x,\xi,h}$ by the family $\chi(x,\xi) \tilde \chi(\eta) k_{x,\xi,h}(\eta)$,
where $\tilde \chi$ is another cutoff function with bounded derivatives supported in $Q_{\epsilon}$, and equal to one on $Q_{\frac{\epsilon}{2}}$.
Since $v_h$ is polynomially bounded as a tempered distribution, we can estimate 
$v_h(\chi(x,\xi) \tilde \chi(\cdot) k_{x,\xi,h}(\cdot))$ by $h^{-M_1} p(\chi(x,\xi) \tilde \chi(\cdot) k_{x,\xi,h}(\cdot))$ for some Schwartz semi-norm $p$. This shows that for $h \in (0,1]$ we have 
$$
 | v_h(\chi(x,\xi) \tilde \chi(\cdot) k_{x,\xi,h}(\cdot))| \leq C_1 h^{-M_2} (1 + |x|^2 + |\eta|^2)^{M_3} e^{-\delta_1 h^{-1}}
$$
for some $\delta_1,C_1,M_2,M_3>0$.
The integral $I_{2,h}$ is just the pairing of $v_h(\chi(x,\xi) \tilde \chi(\cdot) k_{x,\xi,h}(\cdot))$ with the polynomially bounded family
$$
 (T_h f_{1,h})(x_1,\xi_1)\cdots(T_h f_{n,h})(x_{n},\xi_{n})
$$
of test functions in $\mathcal{S}(\R^{2N})$. We therefore obtain an exponentially small integral $I_{2,h}$.
 \end{proof}

For fields satisfying the Wightman axioms with the cluster property and vanishing one-point distribution
the general form of the two-point function is given by the K\"allen-Lehmann representation
$$
 w_2(f_1,f_2) = \int w_{2,m}(f_1,f_2) \der \rho(m),
$$
where $w_{2,m}(x,y)$ is the two-point function of the free scalar field of mass $m \geq 0$, and $\der\rho$ is a polynomially bounded measure supported on $[0,\infty)$ that we refer to as the spectral measure (see for example \cite{MR0493420}*{Theorem IX.34}.
This allows one to compute the analytic wavefront set of $w_2(f_1,f_2)$. If the Fourier transform of the spectral measure is not analytic the analytic wavefront set of $w_2$ can contain timelike vectors. Since one can construct polynomially bounded measures, for which the Fourier transform is not analytic, this shows that the two-point function cannot be expected to contain only lightlike vectors. An example is the spectral measure 
$$
 \der \rho(m) = \begin{cases} 
 e^{-m^\alpha} \der m & m \geq m_0 \\
 0 &  m < m_0
 \end{cases}
$$
for some $m_0>0$ and $0 < \alpha < 1$.
Then the Fourier transform of the measure is a Gevrey function, but is not analytic at $0$. This leads to elements in the analytic wavefront set of the form $(x,-\xi,x,\xi)$, where $\xi$ is future directed and timelike. 
One can construct spectral measures of the form $\sigma(m) \der m$ with rapidly decreasing $\sigma$ such that points
$(x,-\xi,y,\xi)$ occur in the analytic wavefront set where $x \not=y$ is in the interior of the light cone based at $y$, and such that $\xi$ is timelike.

\subsection{The Free Klein-Gordon field as an example}

In this section we will show that under relatively mild assumptions any analytic Hadamard state for the free Klein-Gordon field is in fact tempered analytic. 
We assume here that $M$ is a globally hyperbolic spacetime and we fix a mass $m \geq 0$. Then the Klein-Gordon operator
$\Box + m^2$ admits unique retarded and advanced fundamental solutions $G_{\ret/\av}: C^\infty_0(M) \to C^\infty(M)$.
These maps are continuous and uniquely determined by the properties
\begin{itemize}
 \item $\supp (G_{\ret/\av} f) \subseteq J^\pm(\supp f)$ for any $f \in C^\infty_0(M)$,
 \item $(\Box + m^2)G_{\ret/\av} f =  G_{\ret/\av} (\Box + m^2) f=f$ for any $f \in C^\infty_0(M)$.
\end{itemize}
Here $J^\pm(K)$ is the causal future/past of the set $K \subseteq M$.
It is also convenient to define the map $G = G_\ret - G_\av$, which maps $C^\infty_0(M)$ onto the space of solutions of $(\Box + m^2) f=0$ with space-like compact support, i.e. with support that has compact intersection with any spacelike Cauchy surface. We will denote by $\tilde G \in \mathcal{D}'(M \times M)$ its integral kernel, so that
the distribution $\tilde G(\cdot, f)$ equals $G f$ for all  $f \in C^\infty_0(M)$.

The Klein-Gordon field algebra is the $*$-algebra $\mathcal{A}$ with unit $\mathbf{1}$ generated by symbols $\Phi(f), f \in C^\infty_0(M)$ and relations
\begin{align*}
f &\to \Phi(f)  \textrm{ is linear},\\
 \Phi((\Box + m^2) f)&=0,\textrm{ for all } f \in C^\infty_0(M),\\
 [\Phi(f_1),\Phi(f_2)] & = - \rmi \tilde G(f_1,f_2) \mathbf{1}, \textrm{ for all } f_1,f_2 \in C^\infty_0(M),\\
 \Phi(f)^* &= \Phi(\overline{f}) \textrm{ for all } f \in C^\infty_0(M).
\end{align*}
A state $\omega: \mathcal{A} \to \C$ then defines via the GNS-construction a quantum field theory.

\begin{assumption} \label{assume}
We now fix an embedding $\iota: M \to \R^{d'}$ and make the following assumptions.
\begin{enumerate}
 \item The projection $\mathrm{pr}_1 \circ \iota$ to the first component is a global proper time function $t: M \to \R$ which induces a foliation of $M$ into spacelike Cauchy surfaces, such that $\iota$ equals the projection of $\iota$ to the first component.
 \item $\Box$ extends to a continuous map $\mathcal{S}(M) \to \mathcal{S}(M)$. 
 \item $G$ extends to a continuous map from $\mathcal{S}(M)$ to $\mathcal{S'}(M) \cap C^\infty(M)$. \label{tempcon1}
 \item $G$ maps $\mathcal{S}_{h,O}(M)$ to the set of families of distributions whose microsupport is contained in the zero section. \label{tempcon2}
 \end{enumerate}
\end{assumption}

Condition \eqref{tempcon1} is clearly a temperedness assumption which implies in particular that $\tilde G$ is a continuous bilinear form on $\mathcal{S}(M) \times \mathcal{S}(M)$.
To understand the meaning of \eqref{tempcon2} note that analogous conditions automatically hold for compactly supported test functions. Namely, if $(f_h) \in C^\infty_{0,h,O}(M)$ then $G_{\ret/\av} f_h$
has its microsupport in the zero section. Indeed, this follows from propagation of singularities, \cite{MR1872698}*{Theorem 4.3.7 and Remark 4.3.10}, 
as $(\Box + m^2) G_{\ret/\av} f_h = f_h$ and therefore any non-zero element in the microsupport would propagate away from the support of $f_h$ to the future and the past (
the assumption of a bounded $L^2$-norm in that reference can be replaced by a polynomially bounded $L^2$-norm, by multiplying with an appropriate power of $h$).
The last condition is therefore also a temperedness assumptions on the way singularities propagate.  Checking assumption \eqref{tempcon2} in particular spacetimes would involve showing some kind of uniform analyticity of the Green's function when one of the variables goes to timelike infinity.

One can check that all the conditions are satisfied for the free Klein-Gordon field on Minkowski spacetime if $t$ is chosen as the time-coordinate of an inertial coordinate system. 

We have the following theorem.

\begin{proposition}
Suppose the Assumptions \ref{assume} hold and let $\omega$ be an analytic state for the Klein-Gordon field.
Assume further that for every $n \in \mathbb{N}$ the $n$-point function $\omega_n$ extends to a continuous map $$\omega_n: \mathcal{S}(M) \otimes \ldots \otimes \mathcal{S}(M) \to \C.$$
Then $\omega$ is tempered analytic.
\end{proposition}
\begin{proof}
We split the proof into two steps.\\
{\sl Step 1:}
Assume that $(q_{j,h}) \in \mathcal{S}_{h,O}(M)$. Assume that $f_h \in \mathcal{S}_{h}(M)$ is microlocally uniformly supported in a compact set $K \subset T^*M$ with $K \cap V^- = \emptyset$. 
We need to show that
$$
  \| \Phi(f_h) \Phi(q_{1,h})\cdots \Phi(q_{n,h}) \Omega \| \leq  C e^{-\delta h^{-1}}.
 $$
First we note that there exists a bump function $\chi_K$ which is smooth and compactly supported and equal to one on $K$. Then we can write 
$f_h = \chi_K f_h + (1-\chi_K) f_h$.  By Prop. \ref{multmic} the first term is uniformly microsupported in $K$, and the second term is uniformly exponentially small everywhere, i.e. in Schwartz space. By continuity 
$$
 \| \Phi((1-\chi_K)f_h) \Phi(q_{1,h})\cdots \Phi(q_{n,h}) \Omega \| \leq  C_1 e^{-\delta h^{-1}}.
$$
We can therefore assume without loss of generality that
$f_h \in C^\infty_{0,h,O}(M)$, simply by replacing $f_h$ by $\chi_K f_h$. \\

{\sl Step 2:}
We now proceed by induction in $n$. For $n=0$ we are dealing with the vector
$\Phi(f_h) \Omega$. Since the state is analytic and $f_h \in C^\infty_{0,h,O}(M)$ this implies the estimate.
By induction assume the estimate is correct for $n-1$. We now can write
\begin{align*}
  \Phi(f_h) \Phi(q_{1,h})\cdots \Phi(q_{n,h}) \Omega =  -\rmi \tilde G(f_h,q_{1,h}) \Phi(q_{2,h})\cdots \Phi(q_{n,h}) \Omega +
   \Phi(q_{1,h})  \Phi(f_h) \Phi(q_{2,h})\cdots \Phi(q_{n,h}) \Omega,
\end{align*}
where we have used the relation $[\Phi(f_1),\Phi(f_2)] = - \rmi \tilde G(f_1,f_2) 1$. The second term is exponentially small by the Cauchy-Schwartz inequality, the exponential smallness of the family $ \Phi(f_h)  \Phi(q_{2,h})\cdots \Phi(q_{n,h}) \Omega$, and the fact that $\Phi(q_{1,h})^* \Phi(q_{1,h})  \Phi(f_h) \Phi(q_{2,h})\cdots \Phi(q_{n,h}) \Omega$ is polynomially bounded. The statement then follows if we show that $\tilde G(f_h,q_{1,h})$ is exponentially small, thus establishing the required estimate.  Since the functions 
$q_{2,h},\ldots q_{n,h}$ are no longer required for the argument we will write $q_{h}$ for $q_{1,h}$. We are thus left to establish that $\tilde G(f_h,q_{h})$ is exponentially small for all $(f_h) \in C^\infty_{0,h}(M)$ with the required support properties.
By Prop. \ref{secondchar} it is now sufficient to show that given $q_h \in \mathcal{S}_{h,O}(M)$, $v_{q_h} = G q_h$ has its microsupport contained in $V^+$.
This follows from Assumption \ref{assume}, \eqref{tempcon2}, as in fact the microsupport is contained in the zero-section.
\end{proof}

For the Klein-Gordon field ground and KMS-states on analytic stationary spacetimes are known to be analytic Hadamard states (\cite{SVW}). It has also been shown recently that general analytic globally hyperbolic spacetimes admit analytic Hadamard states (\cite{MR3919442}).

\subsection{Relation to the microlocal spectrum conditions}

The existence of an analytic vector $\Omega$ with $\calA \Omega$ dense in $D$ with respect to the graph topology
has a natural interpretation in terms of analytic microlocal spectrum conditions if the field is constructed in the usual manner from its $m$-point functions. To understand this assume that
we are given a family of $m$-point functions, i.e.
 a family of distributions $(\omega_m)_{m \in \mathbb{N}}, \omega_m \in \mathcal{D}'(M\times \ldots \times M)=\mathcal{D}'(M^m)$ by
$$
 \omega_m(f_1 \otimes \ldots \otimes f_m) =  \langle \Omega, \Phi(f_1) \cdots  \Phi(f_m) \Omega \rangle.
$$
The Wightman reconstruction theorem states roughly that the field theory can be reconstructed from the set of $m$-point functions. This is based on the very general and robust GNS construction that provides a Hilbert space representation for every state on an abstract $*$-algebra.

Spectrum conditions have been postulated in this context for quantum field theory on curved spacetimes.
The introduction of microlocal spectrum conditions in quantum field theory on curved spacetimes started with the realisation by Radzikowski (\cite{MR1400751}) that the Hadamard condition for the two point function of the Klein-Gordon field can be formulated in a microlocal manner, as described by Duistermaat-H\"ormander (\cite{MR388464}) in terms of the wavefront set of the two point function. Since then there were several attempts to find a condition for interacting fields that imposes a similar condition as the spectrum condition for Wightman fields at least microlocally. 

One version of a smooth microlocal spectrum condition was introduced by Brunetti, Fredenhagen, and K\"ohler (\cite{BFK}) as condition on the wavefront set of the $n$-point functions. It is satisfied by any Wightman quantum field in Minkowski spacetime.  The analytic version of this microlocal spectrum condition on real analytic spacetimes was introduced in \cite{SVW} and it was shown that it implies the Reeh-Schlieder property. 
Hollands and Wald in \cite{MR2563799} give a slighlty different condition, allowing for interaction vertices on the spacetime, however requiring the vectors to be lightlike. Their condition excludes certain generalized free fields and is stronger than the condition imposed by  Brunetti, Fredenhagen, and K\"ohler.
The precise form of microlocal spectrum condition is perhaps still not final and may depend on the type of theory one wants to consider. Microlocal spectrum conditions restrict the (analytic) wavefront set of the $m$-point distribution to a closed conic set $\Gamma_m \subseteq T^*M^m \setminus 0$.
\begin{definition}
 A state $\omega$ given by a family of $m$-point distributions is said to satisfy the {\sl analytic microlocal spectrum condition} with respect to $\Gamma_m$
 if $\WFa(\omega_m) \subseteq \Gamma_m$.
\end{definition}
 All proposed microlocal spectrum conditions have the following property in common.
\begin{itemize}
 \item If $(x_1,\xi_1,x_2,\xi_2,\ldots,x_{j-1},\xi_{j},x_{j+1},0,x_{j+1},0,\ldots,x_m,0) \in \Gamma_m$ then $\xi_j$ must either vanish or be future directed and causal.
\end{itemize}
This also implies that $\Gamma_m \cap (-\Gamma_m) = \emptyset$.

As shown in \cite{SVW} any Wightman field theory in Minkowski spacetime satisfying the usual Wightman axioms satisfies the analytic microlocal spectrum condition with respect to certain $\Gamma_m$.

\begin{proposition} \label{Hilbertwavefront}
 Assume the analytic microlocal spectrum condition holds with respect to $\Gamma_m$ satisfying the above condition. Then $\Omega \in D_a$, i.e. $\Omega$ is analytic.
 \end{proposition}

\begin{proof}
 Given a Hilbert space valued distribution $u(\cdot)$ on a real analytic manifold $X$ one can form the complex valued distribution $\langle u(\bar \cdot), u(\cdot) \rangle$ on $X \times X$. The observation (Prop 2.6, 2), Equ. (12) in \cite{SVW}) is that $(x,\xi)$ is in the wavefront set of $u$ if and only if $(x,-\xi,x,\xi)$ is in the wavefront set of $\langle u(\bar \cdot), u(\cdot) \rangle.$ We apply this to the vector-valued distribution $\Phi(\cdot) \cdots \Phi(\cdot) \Omega$.
 Assume that $(x_1,\xi_1,\ldots, x_j,\xi_j,x_{j+1},0,\ldots,x_m,0)$ is contained in the analytic wavefront set of $u$. By the above
   we know that $$(x_m,0,x_{m-1},0,\ldots,0,x_{j-1},0,x_j,-\xi_j,\dots, x_1,-\xi_1, x_1,\xi_1,x_2,\xi_2,\ldots,x_j,\xi_j,x_{j+1},0,\ldots,x_m,0)$$ is in the analytic wavefront set of $\omega_{2m}$.
 By the microlocal spectrum condition the covector $\xi_j$ must be in the closed forward lightcone. 
 \end{proof}

\begin{rem} \label{r412}
 In the same way as analytic vectors are defined one can also define smooth vectors by replacing the analytic wavefront set by the usual smooth wavefront set.  Hence, a vector $\Omega \in D$ is smooth if and only if for all $n \in \mathbb{N}$ 
the first non-zero covector from the right in the wavefront set of the Hilbert space
 valued distribution $\Phi(\cdot) \cdots \Phi(\cdot)  \Omega$ on $M^n$ 
 is future directed and causal.
It is clear that any analytic vector is also smooth.
 It is easy to see that the set of smooth vectors is invariant under the action of fields $\Phi(f), f \in C^\infty_0(M)$.
 For any smooth vector $\phi$ the wavefront set of the Hilbert space-valued distribution $\Phi(\cdot) \phi$ is contained in the lightcone. This has two immediate consequences. Since the wavefront set does not intersect the normal of a timelike curve, the distribution $\Phi(\cdot) \phi$ can be restricted to such a curve. This reproduces results in \cite{MR1885176} in a natural framework. Secondly, this also allows to generalize to curved spacetimes a result by Borchers (\cite{MR192797})  that field operators becomes smooth when smeared out in timelike directions.
 Indeed, let $t$ be a global time function that induces a diffeomorphism $M \cong \R \times \Sigma$. Then, for any smooth vector $\phi$ and any function $h \in C^\infty_0(\R)$ the distribution $f \mapsto \Phi(h \otimes f) \phi$ on $\Sigma$ is a smooth Hilbert space valued function.
This means that smearing the field operator only in the $t$-direction results in a strongly smooth operator function on the domain of smooth vectors.  
This can be seen immediately from the computation rules for wavefront sets summarized in Appendix \ref{comprules}: since the covectors in $T^*M \cong T^*\R \times T^* \Sigma$ of the form $(t,0,y,\eta)$ are not in the wavefront set, the wavefront set of the distribution $\Phi(h \otimes \cdot) \phi$ on $\Sigma$ is empty. 
\end{rem}

\subsection{Internal degrees of freedom}

In the above description the field is scalar-valued. In general one would like a description of fields with spin or of several different types. This can be done by twisting with a vector bundle as follows.
One fixes a complex real analytic vector bundle $E \to M$. 
We assume that $E$ is equipped with a non-degenerate real analytic sesquilinear form that identifies the dual bundle $E^*$ with the complex conjugate bundle $\overline{E}$. Using this identification the bundle $F= E \oplus E^*$ then has a complex conjugation $\bar\cdot$ defined by $\overline{(v,w)} = (w^*,v^*)$. The space of distributions $\mathcal{D}'(M;F)$ taking values in $F$ are identified with the dual of $C^\infty_0(M;F^*)$. The field would then be defined as a map $\Phi: C^\infty_0(M;F^*) \to \mathrm{End}(D)$ with the requirement that $\Phi(f)$ is symmetric on $D$
 if $\overline{f}=f$. The paper can be read in this more general context if one thinks of the bundles as being suppressed in the notation.

\section{The timelike tube theorem} \label{ttt}

Given two points $p,q \in M$ and a smooth timelike curve-segment $\gamma: [0,1] \to M$ connecting $p$ and $q$, we let $I_0(p,q,\gamma)$ be the set of points that can be reached via continuous deformations $\gamma_s$ of smooth timelike curves with fixed endpoints $p,q$. The restriction to smooth curves here is merely for convenience and the definition can also be stated with $C^1$-curves.

\begin{figure}[!htbp] 
 \includegraphics[scale=0.6]{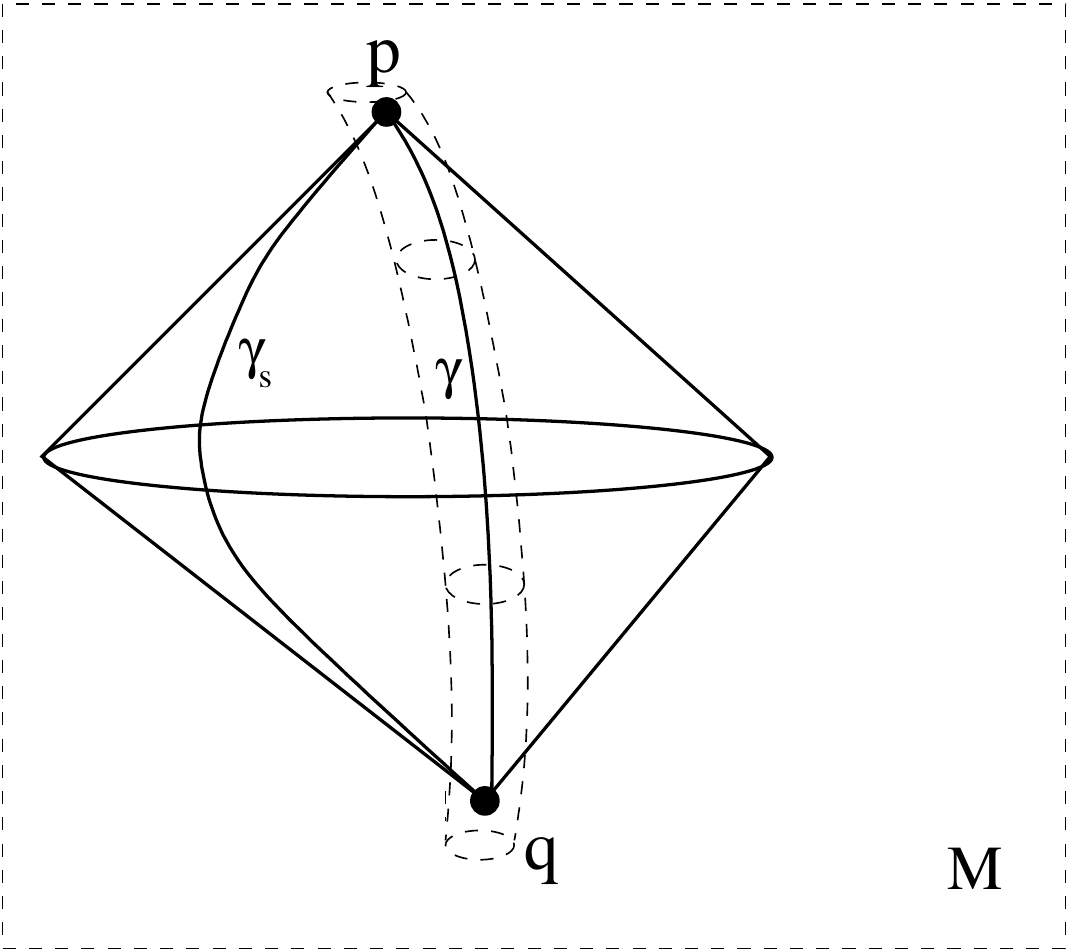}
 \caption{The set $I_0(p,q,\gamma)$ with $\gamma$ contained in a timelike tube.}
\end{figure}

The {\sl timelike tube envelope} $\mathcal{E}_T(\calO)$ of an open set $\calO$ is defined as the smallest set $A$ containing $\calO$ with the property that for any smooth timelike curve-segment $\gamma:[0,1] \to M$  lying in the interior of $A$ with endpoints $p,q$ also the set $I_0(p,q,\gamma)$ is in $A$. 

In addition we can also form that causal envelope $\mathcal{E}_c(\calO)$ as the smallest set $A$ containing $\calO$ such that both $\mathcal{E}_T(\calO)$ and the domain of dependence of $\calO$ are contained in $A$.

\subsection{The Reeh-Schlieder theorem and the timelike tube theorem}\label{tlt}

\begin{theorem} \label{Reeh-Schlieder}
 Suppose that $(\Phi,D \subset \mathcal{H})$ is a quantum field theory and assume that $\Omega$ is an analytic vector such that $\calA \Omega$ is dense in $D$ with respect to the graph topology. Then, for any non-empty open subset $\calO$ the set $\calA(\calO) \Omega$ is dense in $\calH$.
\end{theorem}
\begin{proof}
 It suffices to check that the orthogonal complement of $\calA(\calO) \Omega$ is the zero vector. Assume that
 $\phi$ is a vector in the orthogonal complement of $\calA(\calO) \Omega$. This means the distribution
 $$
  w_m := \langle \phi, \Phi(\cdot) \cdots \Phi(\cdot) \Omega \rangle
 $$
 vanishes on $\calO \times \ldots \times \calO \subset M^m$ for any $m \in \mathbb{N}$. Since $\Omega$ is analytic the wavefront set of this distribution $w_m$ satisfies $\WFa(w_m) \cap - \WFa(w_m) = \emptyset$. By unique continuation, Prop. \ref{ucp},
 this implies that $w_m$ vanishes everywhere. Hence, $\phi$ is orthogonal to the dense set $D$. It follows that $\phi=0$.
\end{proof}

\begin{rem}
 The minor modification of the proof actually shows that a stronger statement holds. Namely, for any countable collection $(\calO_k)_{k \in \mathbb{N}}$ of non-empty open subsets $\calO_k \subset M$ the span of the set 
 $$
  \{\Phi(f_k) \Phi(f_{k-1}) \ldots \Phi(f_1) \Omega \mid\; k \in \mathbb{N} , \quad f_j \in C^\infty_0(\calO_j)\}
 $$
 is dense in $\calH$.
\end{rem}

\begin{theorem} \label{timeliketube}
 Suppose that $(\Phi,D \subset \mathcal{H})$ is a quantum field theory and assume that $\Omega$ is an analytic vector such that $\calA \Omega$ is dense in $D$ with respect to the graph topology. Then, for an open subset $\calO$ let, as above, be $\mathcal{R}(\calO)$ be the local von-Neumann algebra associated with the quantum field. Let $\mathcal{E}_T(\calO)$ be the timelike tube envelope of $\calO$. Then $\mathcal{R}(\mathcal{E}_T(\calO))= \mathcal{R}(\calO)$.
\end{theorem}

Of course, if in addition the time-slice axiom is satisfied we have $\mathcal{R}(\mathcal{E}_c(\calO))= \mathcal{R}(\calO)$.

\begin{proof}[Proof of Theorem \ref{timeliketube}]
 We show that $\mathcal{R}(\calO) = \mathcal{R}(\mathcal{E}_T(\calO))$ by establishing the equality

 \begin{align}
 &\{  A \in \mathcal{L}(\calH) \mid \forall ( f \in C^\infty_0(\calO),\, v,w \in D), \langle w,  A\, \Phi(f) v \rangle = \langle \Phi(\overline{f})  w,  A\, v\rangle \} \nonumber\\ = &
 \{ A \in \mathcal{L}(\calH) \mid \forall  (f \in C^\infty_0(\mathcal{E}_T(\calO)), v,w \in D), \langle w,  A\, \Phi(f) v \rangle = \langle \Phi(\overline{f})  w,  A\, v\rangle \} .\label{eqset}
 \end{align}
 
 Consider an operator $A \in \calL(\calH)$ such that  $\langle w , A \Phi(f) v \rangle = \langle \Phi(\bar f)  w , A v \rangle$
 for all $v,w \in D$ and  $f \in C^\infty_0(\calO)$.
 This is equivalent for the distribution
 $$
  \langle A^* w , \Phi(\cdot) v \rangle - \langle \Phi(\bar \cdot)  w , A v \rangle
 $$
 to vanish on $\calO$ for all $v,w \in D$. This in turn is equivalent to the distributions defined by 
 \begin{align} \label{wequ}
 &w_{m,m'}: h \otimes f_1 \otimes \ldots \otimes f_m \otimes h_1 \otimes \ldots \otimes h_{m'} \mapsto \nonumber \\
 &\langle A^* \Phi(\bar f_1) \cdots \Phi(\bar f_m) \Omega , \Phi(h) \Phi(h_1) \cdots \Phi(h_{m'}) \Omega \rangle - 
 \langle \Phi(\bar h)  \Phi(\bar f_1) \cdots \Phi(\bar f_m) \Omega , A \Phi(h_1) \cdots \Phi(h_{m'}) \Omega \rangle
 \end{align}
 to vanish in the set $\calO \times M^{m+m'}$ for all $m,m' \in \mathbb{N}$.
 We will show below that if $w_{m,m'}$ vanishes in $\calO \times M^{m+m'}$  it automatically vanishes in $\mathcal{E}_T(\calO) \times M^{m+m'}$, thus showing the equality \eqref{eqset}.
 
 Let $S \subset M$ be any co-dimension one timelike hypersurface in $M$ then $S \times M^{m+m'}$ is a hypersurface in $M^{m+m'+1}$ with its conormal at $(x,x_1,\ldots,x_{m},y_1,\ldots,y_{m'}) \in S \times M^{m+m'}$ spanned by $(x, \xi,x_1,0,\ldots,x_{m},0,y_1,0,\ldots,y_{m'},0)$, $\xi$ being any non-zero spacelike co-normal vector to $x \in S$.
 We will show below that the conormal of $S \times M^{m+m'}$ does not intersect $\WFa(w_{m,m'})$.
 By Prop. \ref{ucphyp} we have unique continuation for $w_{m,m'}$ across any such
$S \times M^{m+m'}$. We will now use a deformation argument to show that that $w_{m+m'}$ vanishes on $\mathcal{E}_T(\calO) \times M^{2m}$. 
{We argue by contradiction. Assume that there exists a point $(x_0,y_0) \in \mathcal{E}_T(\calO) \times M^{2m}$
in the support of $w_{m+m'}$. Then there exists a continuous family $\gamma_s$ of smooth timelike curves such that
$\gamma_s(0)=q, \gamma_s(1)=p$, such that $\gamma_0$ is in $\calO$, and such that $\gamma_{s_0}$ passes through the point $x_0$. We introduce a complete Riemannian metric $\tilde g$ on $M$.
Since $K=\{\gamma_s(t) \mid (s,t) \in [0,s_0] \times [0,1] \}$ is compact so is the closure of the set 
$K_1 =\cup_{x \in K} \overline{B_1(x)}$, where $\overline{B_1(x)}$ is the closed unit ball centered at $x$ with respect to the complete metric $\tilde g$. 
Let $N \gamma_s$ be the normal bundle of $\gamma_s$ in $TM$. The bundle of balls of radius $r$ in $N \gamma_s$ will be denoted by $N_r \gamma_s$. Its boundary $\partial N_r \gamma_s$  is the bundle of spheres $S_r \gamma_s$ in the normal bundle. 
For sufficiently small $r>0$ the exponential map of the metric $\tilde g$ then is a local diffeomorphism 
from $N_r \gamma_s$ to a neighborhood of $\gamma_s((0,1))$ for every $s \in [0,s_0]$. 
We now choose $0<\delta<1$ and $0 < \delta' <1$ such that the following hold.
\begin{enumerate}
 \item $\forall t \in [0,1] \forall \tilde s \in [0,\delta'], B_\delta(\gamma_{\tilde s}(t)) \in \calO$, 
 \item for every $s \in [0,s_0]$ the exponential map defines an immersion $S_\delta \gamma_s \to M$ with timelike image.
\end{enumerate}
Such a choice of $\delta$ exists uniformly in $s \in [0,s_0]$ by compactness of $K_1$.
The image of $N_\delta \gamma_s$ then defines a family of tubes $\mathcal{T}_s$ with boundaries $\partial \mathcal{T}_s$ that form a continuous family of immersed timelike hypersurfaces. By assumption the open set 
$I=\{s \in \R \mid \exists (x,y) \in \mathrm{supp}(w_{m+m'}), x \in \mathcal{T}_s \}$ contains $T$ and is bounded below by $\delta'>0$. Let $s_0 = \inf I$. Then $\mathcal{T}_{s_0} \times M^{2m}$ does not intersect the support of $w_{m+m'}$, but its boundary $S \times M^{2m} = \partial \mathcal{T}_{s_0} \times M^{2m}$ does. If there is unique continuation across $S \times M^{2m}$ we obtain a contradiction.
}

It now only remains to show that indeed the conormal of $S \times M^{m+m'}$ does not intersect $\WFa(w_{m,m'})$ as claimed. We must prove that $(x, \xi,x_1,0,\ldots,x_{m},0,y_1,0,\ldots,y_{m'},0)$ with non-zero spacelike $\xi$ is not in the wavefront set of $w_{m,m'}$. By Prop. \ref{altcharanwf} the Hilbert space valued distribution
$\Phi(h) \Phi(h_1) \cdots \Phi(h_{m'}) \Omega $  does not have $(x, \xi,y_1,0,\ldots,y_{m'},0)$ in its wavefront set.
Similarly, the Hilbert space distribution  $\Phi(\bar h)  \Phi( \bar f_1) \cdots \Phi( \bar f_m) \Omega $ taking values in the conjugate vector space does not have $(x, -\xi,x_1,0,\ldots,x_{m},0)$ in its wavefront set. Essentially by the Cauchy-Schwarz inequality we have therefore (see Prop 2.6, 2), Equ. (13) in \cite{SVW}) that $(x, \xi,x_1,0,\ldots,x_{m},0,y_1,0,\ldots,y_{m'},0)$ is not in the wavefront set of either term in \eqref{wequ} and therefore not in the wavefront set of $w_{m,m'}$.   
 \end{proof}

\begin{rem}
The above proof shows that requiring the analytic wavefront set to be contained in the set of lightlike vectors rather than causal vectors implies the time-slice axiom. This is for example the case in generalized free theories if the spectral measure $\rho$ on $[0,\infty)$
is exponentially decaying. This is consistent with the conditions on the spectral measure for generalized free fields to satisfy the time-slice axiom identified in \cite{MR376001}. Even though the time-slice axiom is expected to hold also for reasonable interacting quantum field theories the wavefront set of the $m$-point functions will in general be expected to contain time-like vectors. 
\end{rem} 

{
\section{The algebra of a timelike curve and restrictions of the field operator}
Given any subset $S \subset M$ such as a hypersurface one can define the algebra of observables $\mathcal{R}(S)$
as the intersection $\bigcap\limits_{\calO \supset S} \mathcal{R}(S)$ of the field algebras of all open subsets containing $S$. Given a timelike curve $\gamma: [0,1] \to M$ let us denote by $\gamma^\circ$ the image $\gamma((0,1))$ of the corresponding open interval, i.e. the curve without its endpoints.
One consequence of the timelike tube theorem is that for a timelike curve $\gamma^\circ$ the corresponding
algebra $\mathcal{R}(\gamma^\circ)$ is that of the open set $\calE(\gamma^\circ)$ obtained by deforming the curve in a timelike manner, with the fixed endpoints removed.}
{
There is however another way to define the algebra of observables of a timelike curve based on the fact that the field operator and its derivatives can in fact be restricted to the curve in the sense explained in Remark \ref{r412}.
Namely, on can consider what is essentially the double commutant of the fields restricted to timelike curves. To be more precise, recall that $A$ commutes with the field operator in $\calO$ if and only if
$$
  \langle A^* w , \Phi(\cdot) v \rangle - \langle \Phi(\bar \cdot)  w , A v \rangle
$$
vanishes in $\calO$ for all vectors in the domain $D \subset \calH$.
One can now define the algebra $\mathcal{R}^\infty(\gamma^\circ)$ as the commutant of the set of operators
$A$ with the property that 
$$
  \langle A^* w , \Phi(\cdot) v \rangle - \langle \Phi(\bar \cdot)  w , A v \rangle
$$
is flat on $\gamma$. Here a distribution is called flat on $\gamma$ if its wavefront set does not intersect the normal bundle of $\gamma$ and the restriction of of all its derivatives to $\gamma$ vanish. It was conjectured in \cite{MR1885176} that the two definitions give the same local algebra, i.e. $\mathcal{R}^\infty(\gamma^\circ)= \mathcal{R}(\gamma^\circ)$. Our arguments imply this conjecture if $\gamma$ is analytic. 
\begin{theorem}
 Assume that $\gamma:[a,b] \to M$ is a real analytic timelike curve. Then $\mathcal{R}^\infty(\gamma^\circ)= \mathcal{R}(\gamma^\circ)$.
\end{theorem}
\begin{proof}
The proof is based on the proof of Theorem \ref{timeliketube} and we will use the notation from this proof.
A bounded operator $A$ is in the commutant of $\mathcal{R}^\infty(\gamma^\circ)$ if the distribution $w_{n,m'}$ in the proof of Theorem \ref{timeliketube} is flat on $\gamma^\circ \times M^{m+m'}$. By a theorem of Boman \cite{MR1194524} (see also \cite{MR1382568} for a generalization to the Denjoy-Carleman class) flatness of the distribution implies its vanishing in an open neighborhood of 
$\gamma^\circ \times M^{m+m'}$. Hence, $A$ is in the commutant of $\mathcal{R}(\gamma^\circ)$. Of course
a distribution vanishing near $\gamma^\circ$ is flat on $\gamma^\circ$ and therefore $\mathcal{R}^\infty(\gamma^\circ)$ and $\mathcal{R}(\gamma^\circ)$ have the same commutant.
\end{proof}}

 \begin{appendix} 
 
 \section{Wavefront sets and microlocal concepts} \label{sec2}

\subsection{The wavefront set.}

As usual the space of distributions $\mathcal{D}'(\R^d)$ is the topological dual of the space of compactly supported smooth test functions $C^\infty_0(\R^d)$. We will use the formal notation
$$
 (u,\varphi) = \int_{\R^d} u(x) \varphi(x) \der x
$$
for the distributional pairing if $u \in \mathcal{D}'(\R^d)$ and $\varphi \in C^\infty_0(\R^d)$.
 The space of tempered distributions $\mathcal{S}'(\R^d)$ is the dual of the Schwartz space $\mathcal{S}(\R^d)$. Finally recall that a distribution of compact support can be paired with a test function in $C^\infty(\R^d)$ and this identifies the dual of  $C^\infty(\R^d)$ with the space of compactly supported distributions.

The Fourier transform can be used to analyse distributions locally in phase space. We recall some basic properties of the Fourier transform.
If $u \in \mathcal{D}'(\R^d)$ is a compactly supported distribution then its Fourier transform is a holomorphic function in the complex plane of uniform exponential type. Such a distribution is a smooth function if and only if its Fourier transform is decaying faster than any power. This means
$$
 u \in C^\infty_0(\R^d) \Leftrightarrow \forall N \in \mathbb{N}, \sup_{\xi \in \R^d} |\xi|^N |\hat u(\xi) | < \infty.
$$
This statement can of course be localized by saying that a general distribution $u \in \mathcal{D}'(\R^d)$ is smooth near the point $x_0$ if and only if there exists a function $\chi \in C^\infty_0(\R^d)$ with
$\chi(x_0)=1$ such that
$$
 \sup_{\xi \in \R^d} |\xi|^n |\widehat{\chi \cdot u}(\xi)| < \infty
$$
for any $n \in \mathbb{N}$.
The purpose of the cut-off function is to turn $u$ into a distribution of compact support and localize the statement at the point $x_0$.
The notion of wavefront set seeks to refine the notion of singular support, i.e. the set of points where the distribution is not smooth. It singles out the directions in which a distribution may fail to be smooth by analysing in which directions in the $\xi$ Fourier transform of $\chi \cdot u$ fails to decay faster than any power. 
\begin{definition}
 The wavefront set $\mathrm{WF}(u)$ of a distribution $u \in \mathcal{D}'(\R^d)$ is the complement of the set of points
 in $(x_0,\xi_0) \in \R^d \times \R^d \setminus \{0\}$ such that  there exists a cut-off function $\chi \in C^\infty_0(\R^d)$ with
$\chi(x_0)=1$ and an open conic neighborhood $\Gamma$ of $\xi_0$ such that
$$
 \sup_{\xi \in \Gamma} |\xi|^n |\widehat{\chi \cdot u}(\xi)| < \infty
$$
for all $n \in \mathbb{N}$.
\end{definition}
This additional asymptotic localisation of the singular support property is sometimes called microlocalisation.
One of the key observations relevant to general relativity and geometry is that the wavefront set transforms like a subset of the cotangent bundle under smooth change of coordinates. Hence, the notation of the wavefront set makes sense on smooth manifolds in the absence of a global and invariantly defined Fourier transform.
It is customary to introduce a semi-classical parameter $h>0$ as an energy scale and reformulate this in terms of the semi-classical Fourier transform $\mathcal{F}_h$, which is defined as
\begin{equation}
 (\mathcal{F}_h f)(\xi) = \frac{1}{(2 \pi h)^\frac{d}{2}} \int_{\R^d} f(x) e^{-\frac{\rmi}{h} x \xi} \der x.
\end{equation}
Therefore, $(x_0,\xi_0) \in \R^d \times \R^d \setminus \{0\}$ is not in the wavefront set of $u \in \mathcal{D}'(\R^d)$ if and only if for some cut-off function $\chi \in C^\infty_0(\R^d)$, 
with $\chi(x_0)=1$ we have for each $n \in \mathbb{N}$ a constant $C_n>0$ such that
 $$
  |\mathcal{F}_h(\chi \cdot u)(\xi)| < C_n h^n
$$
for all $n \in \mathbb{N}$, $h \in (0,1]$, and all $\xi$ in an open neighborhood of $\xi_0$. This is just a reformulation of the above, where the semi-classical parameter $h$ is being used to rescale the momentum $\xi$. This definition is sometimes more flexible and also allows to define the wavefront set of an $h$-dependent distribution. This is sometimes called the frequency set or semi-classical wavefront set $\WF_h$. This is introduced for families of distributions that are polynomially bounded in $h$. More precisely, we say an $h$-dependent family of tempered distributions $(u_h)$
is polynomially bounded if there exists a  continuous semi-norm $p: \mathcal{S}(\R^d) \to [0,\infty)$ such that for some $N>0$ we have $u_h(f) \leq h^{-N} p(f)$ for all $f \in \mathcal{S}(\R^d), h \in (0,1]$. The semi-classical wavefront set $\WF_h(u_h)$ is then defined as the complement in $\R^d \times \R^d$ of the set of points $(x_0,\xi_0)$ such that there exists a cut-off function $\chi \in C^\infty_0(\R^d)$, 
with $\chi(x_0)=1$ so that we have for each $n \in \mathbb{N}$ a constant $C_n>0$ with
 $$
  |\mathcal{F}_h(\chi \cdot u_h)(\xi)| < C_n h^n
$$
for all $n \in \mathbb{N}$, $h \in (0,1]$, and all $\xi$ in an open neighborhood of $\xi_0$. For an $h$-independent family we have $\WF_h(u) = \mathrm{supp}(u) \times 0 \cup \WF(u)$.

 \section{Analytic wavefront sets} \label{appa}

Recall from the main body of the text that we define the analytic wavefront set of a distribution $u \in \mathcal{D}'(M)$ 
as the elements in $T^*M \setminus 0$ of the microsupport of $(u_h)=(u)$, regarded as a family of distributions independent of $h$.
For the analytic wavefront set $\WFa$ there are several other equivalent definitions, which we now review. Each has their own advantages and disadvantages.  In the following we assume that $u \in \mathcal{D}'(\R^d,\calH)$ is a distribution taking values in a Hilbert space $\calH$. We will suppress the $\calH$ to keep the notations short. All integrals will need to be understood as distributional pairings, and we will use this convention without further explanation.

\subsection{The classical FBI transform}
One can use another localized version of the Fourier transform $\mathcal{T}_{a} u$, which is also sometimes referred to as the FBI transform. 
It is defined as
$$
 \mathcal{T}_{a} u(x,\xi) = \int_{\R^d} e^{-\frac{a}{2}(x-y)^2} u(y) e^{-\rmi \xi (y-x)} \der y.
$$
This only differs from its semi-classical counterpart by an irrelevant pre-factor and the lack of the extra scaling in the $\xi$-variable.
The distribution $u$ can also be recovered from the transforms $\mathcal{T}_{\xi}$ of $u$ and $x u$. This is done by the 
inversion formula (see \cite{MR1065136}*{(9.6.7)} or \cite{MR3444405}*{(3.36)} in the one dimensions case)
$$
 u(y) = \frac{1}{(2 \pi)^d} \int_{\R^d} (\mathcal{T}_{|\xi|} u)(y,\xi) \der \xi + \frac{1}{(2 \pi)^d} \int_{\R^d}  \frac{1}{|\xi|} \xi \cdot  (\mathcal{T}_{|\xi|} f)(y,\xi) \der \xi.
$$
Here $f(x)$ is the vector-valued function $f(x) = -\frac{1}{2}\rmi (x-y) u(x)$.

\begin{proposition}\label{A1}
 A vector $(x_0,\xi_0) \in \R^d \times \R^d \setminus \{0\}$ is not in the analytic wavefront set if and only if the following holds: there exists a cutoff function $\chi \in C^\infty_0(\R^d)$ which is equal to one near $x_0$ and an open conic neighborhood $\calO$ of $(x_0,\xi_0)$ in $\R^d \times \R^d \setminus \{0\}$ such that for some $C,\delta>0$ we have
 $$
  \| (\mathcal{T}_{|\xi|} (\chi u))(x, \xi ) \| \leq C e^{-\delta |\xi |}.
 $$ 
\end{proposition}

The condition of Prop. \ref{A1}. is satisfied if and only if the semi-classical FBI-transform
$\TT_h(u)(x,\xi)$ is exponentially decaying in $h^{-1}$ for all unit vectors $\xi$ in an open neighborhood of $|\xi_0|^{-1} \xi_0$ in the unit sphere.  The proof of the equivalence can be found in \cite{MR1065136}*{Section 9.6, in particular Theorem 9.6.3} with a slighlty different notation. 

\subsection{H\"ormander's approach}

The following provides a criterion that is used by H\"ormander (\cite{MR1065136}) as a definition of the analytic wavefront set. This is also the approach favored in \cite{SVW} in the analysis of quantum field theory and the Reeh-Schlieder property.

\begin{proposition}\label{A2}
 A vector $(x_0,\xi_0) \in \R^d \times \R^d \setminus \{0\}$ is not in the analytic wavefront set if and only if the following holds: there exists an open neighborhood $\calO$ of $x_0$, an open conic neighborhood $\Gamma$ of $\xi_0$, and a bounded sequence $u_n$ of distributions with compact support that equals to $u$ in $\calO$ such that for some $C>0$ we have
 $$
   |\xi|^n \| \hat u_n(\xi) \| \leq C( C (n+1))^n,
 $$ 
 for all $\xi \in \Gamma$.
 \end{proposition}
 One can choose the sequence $u_N$ always in such a way that $u_n = \chi_n u$, where $\chi_n$ is a sequence of suitably chosen cut-off functions that are equal to one on $\calO$ and with the property that for each multi-index $\alpha$ there exists a constant $C_\alpha>0$ such that
$$
 |(\partial^{\alpha + \beta} \chi_n)(x) | \leq C_\alpha( C_\alpha (n+1))^{|\beta|}.
$$
Such a sequence always exists. 
The sequence $\chi_n$ plays a similar role here to that of the Gaussian factor in the FBI-transform and is used to provide an $n$-dependent localisation.

\subsection{Boundary values of analytic functions}

Another characterisation of the analytic wavefront set is based on boundary values of analytic functions.
We state this here as in \cite{MR3444405}*{Th. 3.38} but refer to \cite{MR1065136}*{Section 8.4, and Th. 8.4.15} for proofs.

 \begin{proposition}\label{A3}
 A vector $(x_0,\xi_0) \in \R^d \times \R^d \setminus \{0\}$ is not in the analytic wavefront set if and only if the following holds: there exists an open complex neighborhood $W$ of $x_0$ in $\C^d$, and open convex convex cones $\Gamma_1,\ldots,\Gamma_k$ in $\{y \in \R^d \mid y \cdot \xi_0 < 0\}$,  functions $u_1,\ldots,u_k$, where $u_j$ is holomorphic in $W \cap \R^d + \rmi \; \Gamma_j$ so that $u$ can be written as the sum of distributional boundary values of the $u_j$. 
 \end{proposition}

For tempered distributions on $\R^d$ there is a direct way to achieve the above decomposition and give an explicit criterion for the analytic wavefront set. To understand this let us discuss this in one dimension first:\\
The function $K(z) = \frac{1}{4} \mathrm{sech}(\frac{\pi}{2} z) = \left( 4 \cosh(\frac{\pi}{2} z) \right)^{-1}$ is holomorphic in the strip $|\Im(z)|<1$ and decays exponentially along the real axis. We have
$\int K(x + \rmi y) \der x = \pi$ for all $|y|<1$. Moreover, the family $\frac{1}{\pi} \Re K(x + \rmi y)$ is a $\delta$-family as $y \to \pm 1$. Given a tempered distribution $u \in \mathcal{S}'(\R)$ we can convolve with $K$ to obtain a function
$$
 K_u(z) = \int K(z-x) u(x) \der x = (K*u)(z)
$$
which is holomorphic in the strip $\Im(z)<1$. Moreover, 
$$
 u(x) = \lim_{\epsilon \searrow 0}\left( K_u(x + \rmi - \rmi \epsilon) + K_u(x - \rmi + \rmi \epsilon) \right) = u_+ + u_-.
$$
This decomposes $u$ as the sum of two distributions, that are one sided boundary values of holomorphic functions.
Therefore, if $\xi>0$ then $(x,\xi) \notin \WFa(u)$ if and only if $u_+$ is real analytic near $x$. Similarly, if
$\xi<0$ then $(x,\xi) \notin \WFa(u)$ if and only if $u_-$ is real analytic near $x$.

This construction can be generalized to higher dimensions as follows. One defines
\begin{align*}
 I(\xi) &= \int_{|\omega|=1} e^{-\langle \omega, \xi \rangle} \der \omega,\\
 K(z) &= \frac{1}{(2 \pi)^d} \int_{\R^d} e^{ \rmi \langle \xi, z \rangle} / I(\xi) \der \xi.
\end{align*}
The analogue of the above decomposition is
$$
 u(x) = \int_{|y|=1} K_u(x + \rmi y) \der y,
$$
where as before $K_u = K * u$. The function $K_u(z)$ is holomorphic in $|\Im(z)| <1$. 

\begin{proposition} Let $u \in \mathcal{S}'(\R^d)$. Then, given a direction $y \in \R^d$ with $|y|=1$,  we have that $(x,y) \notin \WFa(u)$ if and only if $K_u$ is complex analytic in $x+ \rmi y$ near $(x,y)$.
\end{proposition}
 
The above statements and proofs can be found in \cite{MR1065136}*{Theorem 8.4.11}.

 \section{Computational rules for the analytic wavefront set} \label{comprules}
 
Here we summarize some computational rules for analytic wavefront sets. We will be mostly concerned with distributions defined on an arbitrary analytic manifold $M$, but occasionally will also encounter different dimensionalities. In the latter case
we let $M_1, M_2$ be different analytic manifolds of dimensions $d_1$ and $d_2$, respectively.

We start with some simple rules that are easy to see directly from the definitions. Given $u,v \in \mathcal{D}'(M)$,
$$
 \WFa(u + v) \subseteq  \WFa(u) \cup \WFa(v). 
$$
Given $u_1 \in \mathcal{D}'(M_1)$ and $u_2 \in \mathcal{D}'(M_2)$  the tensor product 
$u_1 \otimes u_2$ is a distribution on $M_1 \times M_2$. We have
$$
  \WFa(u_1 \otimes u_2) \subseteq (\WFa(u_1) \times 0) \cup (0 \times \WFa(u_2)) \cup (\WFa(u_1) \times \WFa(u_2)).
$$
Given a smooth map $f: M_1 \to M_2$ the conormal $N_f$ of $f$ is defined as the set of $(x,\xi) \in T^* M_2$
such that $x$ is in the range of $f$ and the pull-back $f^*(\xi)$ vanishes.

Pull-backs and products of distributions may not necessarily be well defined. The wavefront set provides a useful criterion for pull-backs and products to exist in a reasonable way.
If $N_f \cap \mathrm{WF}(u) = \emptyset$ then the pull-back $f^* u$ exists. If $f$ is real analytic then
$$
 \WFa(f^* u) \subseteq f^*(\WFa(u)).
$$
This means the wavefront set has good functorial properties under analytic maps. The condition $N_f \cap \mathrm{WF}(u) = \emptyset$ simply makes sure that $f^*(\mathrm{WF}(u))$ does not intersect the zero section.
Note that this condition is always satisfied if $f$ is an analytic submersion.

The product $u v$ of two distributions $u,v \in \mathcal{D'}(M)$ is not always well defined. There is a similar restriction on the wavefront sets to ensure existence of a product. Namely, the product is well defined as a distribution if
$\mathrm{WF}(u) + \mathrm{WF}(v)$ does not intersect the zero section in $T^*M$. In that case
\begin{align*}
 \mathrm{WF}(u v) &\subseteq \mathrm{WF}(u) + \mathrm{WF}(v),\\
 \WFa(u v) &\subseteq \WFa(u) + \WFa(v).
\end{align*}
Again, the restriction makes sure there are no zero-vectors appearing on the right hand side of this formula.
The formula for the product is actually a special case of the formula for pull-back and tensor products.
This is because the product $u v$ can be seen as the pull back of $u \otimes v$ under the diagonal embedding
$x \mapsto (x,x)$, whose co-normal consists of covectors of the form $(x,\xi,x,-\xi)$.

If $P$ is a partial differential operator on $M$ with analytic coefficients and $P u = f$, then
$$
 \WFa(u) \subseteq \WFa(f) \cup \mathrm{char}(P),
$$
where $\mathrm{char}(P) \subset T^*M \setminus 0$ is the characteristic set of $P$, i.e. the zero set of the principal symbol of $P$.

The above statements are summaries of results that can be found in H\"ormander's book \cite{MR1065136} in Sections 8 and 9.

Another operation is to formally integrate with respect to a subset of variables or smear a subset of variables with respect to test functions. If $u \in \mathcal{D}'(M_1 \times M_2)$ and $f \in C^\infty_0(M_2)$ we can define a distribution
$u_1 \in \mathcal{D}'(M_1)$ by $u_1(\cdot) = u(\cdot \otimes f)$. Formally one can write this as
$$
 u_1(x) = \int_{M_2} u(x,y) f(y) \der y.
$$
We have then
$$
 \mathrm{WF}(u_1) \subseteq \{(x,\xi) \in T^*M_1 \setminus 0 \mid (x,\xi,y,0) \in \mathrm{WF}(u) \textrm{ for some } y \in M_2 \}.
$$
Such a statement is certainly not correct for the analytic wavefront set when $f$ is a general compactly supported test function. When $u$ is compactly supported the above holds for the analytic wavefront set in case $f$ is a real analytic test function  \cite{MR1065136}*{Th. 8.5.4}. 
This statement is however not true in the context of tempered distributions on $\R^d$.
Indeed, the function
$$
 u(x,y) = \frac{1}{x^2 + \frac{1}{y^2+1}}
$$
defines a tempered distribution. Since it is a real analytic function its analytic wavefront set is empty.
Pairing in the $y$-variable with the real analytic Gaussian $e^{-\frac{1}{2}y^2}$ gives the function
$$
 g(x) = \int_\R \frac{1}{x^2 + \frac{1}{y^2+1}} e^{-\frac{1}{2}y^2} \der y.
$$
The derivatives at zero are therefore given by
\begin{align*}
 g^{(k)}(0) &=  -\frac{1}{2} \int_{\R} e^{-\frac{y^2}{2}} \sqrt{-y^2-1} k!
   \left(\left(-y^2-1\right)^{\frac{k+1}{2}}-\left(-\sqrt{-y^2-1}\right)^{k+1}\right) \der y\\
   &=  \sqrt{\pi }  \cos \left(\frac{\pi 
   k}{2}\right) \Gamma (k+1)\, U\left(\frac{1}{2},\frac{k+5}{2},\frac{1}{2}\right),
\end{align*}
where $U$ is the Tricomi confluent hypergeometric function. The radius of convergence of the Taylor series at zero is therefore zero and the function is not real analytic. The mechanism in this example is that the radius of convergence of the Taylor series in the $x$-variable is not positive uniformly in the $y$-variable.

 \section{Unique continuation and the microlocal Holmgren uniqueness theorem}
 
 Let $S \subset M$ be a codimension one hypersurface. If $x \in S$ then there exists a small neighborhood $\U$
 of $x$ such that $S$ divides $\U$ into two connected components $\U_-$ and  $\U_+$.
 For a class of distributions $\V \subset \mathcal{D}'(M)$ we say that we have unique continuation across $S$ in $\U$ if whenever $u \in \V$ vanishes in $\U_-$ it also vanishes in a neighborhood of $S$ in $\U$. 
We say unique continuation holds across $S$ if for every point $x \in S$ there exists a neighborhood $\U$ such that
unique continuation holds across $S$ in $\U$. This is clearly a local definition.

If unique continuation holds across $S$ for distributions $u \in \V$ this means that the support of $u$ cannot touch $S$ from only one side. The following two propositions follow from the very general statement of \cite{MR1065136}*{Th. 8.5.6'}, but we discuss them here separately. 
The first is a microlocal generalisation of Holmgren's uniqueness theorem for partial differential equations with analytic coefficients.

\begin{proposition} \label{ucphyp}
 Suppose that $S \subset M$ is a codimension one hypersurface. Then unique continuation holds across $S$ for all distributions whose analytic wavefront set $\WFa$ does not intersect the conormal bundle of $S$ in $M$.
\end{proposition}

In particular this implies that in case $M$ is a spacetime we have unique continuation across all timelike hypersurfaces for the class of distributions whose analytic wavefront set is contained in the light-cone. If a distribution contains only lightlike covectors in its analytic wavefront set we have in addition unique continuation across any spacelike hypersurface. 

The second is a microlocal version of a form of the edge of the wedge theorem.\begin{proposition} \label{ucp}
 Unique continuation holds across any hypersurface for the class of distributions $u \in \mathcal{D}'(M)$ satisfying
 $\WFa(u) \cap -\WFa(u) = \emptyset$. In particular, if $M$ is connected and such a $u \in \mathcal{D}'(M), \WFa(u) \cap -\WFa(u) = \emptyset$ vanishes on a non-empty open subset of $M$, then $u$ must vanish identically.
\end{proposition}

We have used here that unique continuation across any hypersurface implies in particular unique continuation across any ball. By a standard argument this implies that the support is open and closed and therefore the unique continuation property.

\section{Properties of microlocally uniformly exponentially small families}

In this section we provide some statement about uniform microlocalisation that were used in the main part.

Recall that a polynomially bounded family $(u_h) \in \mathcal{S}_h(\R^d)$ of Schwartz functions is uniformly microsupported in $K \subset \R^d \times \R^d$ if we have that for every $\epsilon>0$ there exists a $\delta>0$ such that for any $N>0$ there is a constant $C_N>0$ with
$$
 \| (1+|x|^2 + |\xi|^2)^N T_h u_h(x,\xi) \| \leq C_N e^{-\delta h^{-1}}, \; \textrm{ for all }{(x,\xi) \in \R^d \times \R^d} \textrm{ with } \mathrm{dist}((x,\xi),K)\geq\epsilon.
$$

We first remark that this definition is equivalent to the $T_h f_h$ being exponentially small in $\mathcal{S}(\R^d \times \R^d)$ after localisation to the complement of $K$. To explain this let $K_\epsilon$ be the set of points $\{ (x,\xi) \in \R^d \times \R^d \mid \mathrm{dist}((x,\xi),K) < \epsilon\}$.

\begin{proposition} \label{swchar}
 The family $(u_h) \in \mathcal{S}_h(\R^d)$ is uniformly microsupported in $K \subset \R^d \times \R^d$ if and only if 
 for any cut-off function $\chi \in C^\infty_0(\R^d \times \R^d)$ that equals one on $K_\epsilon$ we have that
 $(1-\chi) T_h u_h$ is exponentially small in $\mathcal{S}_h(\R^d \times \R^d)$ in the sense that there exists a $\delta>0$
 such that $p((1-\chi) T_h u_h) = O(e^{-\delta h^{-1}})$ for any Schwartz semi-norm $p$. 
\end{proposition}
\begin{proof}
 The statement about the Schwartz semi-norms already includes the estimate in the definition of being uniformly microsupported in $K$. We therefore need to show the corresponding bound for the other Schwarz-seminorms, i.e. derivatives in $x$ and $\xi$.
 The proof is based on the fact that $e^{\frac{\xi^2}{2h}}(T_h u)(x,\xi)$ is holomorphic in $z=x-\rmi \xi$ and we apply the usual principle that derivatives of holomorphic functions at a point can be bounded by localized $L^\infty$-norms. We therefore denote
 $R_h(z) = e^{\frac{\xi^2}{2h}}(T_h u_h)(x,\xi)$, where $z=x-\rmi \xi$ is holomorphic. The bound we have is therefore
 $$
  \| \tilde R_h(z) \| \leq \frac{C_N}{(1+|z|^2)^N} e^{\frac{\Im(z)^2}{2h}} e^{-\delta h^{-1}}
 $$
 outside $K_{\frac{\epsilon}{2}}$.
The bound we seek is
 $$
  \| \tilde \partial^{\alpha}_z \partial^{\beta}_{\overline z} R_h(z) \| \leq \frac{C_N}{(1+|z|^2)^N} e^{\frac{\Im(z)^2}{2h}} e^{-\delta h^{-1}}.
 $$
 outside $K_\epsilon$.
 This is however immediate by differentiating the Bochner-Martinelli formula applied to a small enough polydisc. Note that the left hand side vanishes unless $\beta=0$.
 We then obtain
 $$
  \| \tilde \partial^{\alpha}_z \partial^{\beta}_{\overline z} R_h(z) \| \leq \frac{C'_N}{(1+|z|^2)^N} e^{\frac{\Im(z)^2}{2h}} e^{-\delta' h^{-1}}.
 $$
 outside of $K_\epsilon$. Localisation with the test function gives extra terms in the derivative that can be bounded by repeated application of the product rule.
\end{proof}

For $\mu>0$ define a modified FBI-transform $T_{\mu,h}$, changing the standard deviation in the Gaussian weight,
$$
 (T_{\mu,h} u)(x,\xi) = \mu^\frac{d}{4} 2^{-\frac{d}{2}} (\pi h)^{-\frac{3d}{4}}\int_{\R^d} e^{-\frac{\mu}{2 h}(x-y)^2} u(x) e^{-\frac{\rmi}{h} (x-y) \cdot \xi} \der x = \mu^{-\frac{d}{2}} T_{\mu^{-1} h}(x, \mu^{-1} \xi),
$$
\begin{lemma} \label{hilfe}
If $(u_h) \in \mathcal{S}_h(\R^d)$ is uniformly microsupported in $K \subset \R^d \times \R^d$ this implies that for fixed $\mu>0$ and
every $\epsilon>0, N>0$ there exist $C,h >0$ such that
$$
 \|(1+ |x|^2 + |\xi|^2)^N T_{\mu,h} u_h(x,\xi) \| \leq C e^{-\delta h^{-1}}, \quad \textrm{ for all }{(x,\xi) \in K} \textrm{ with } \mathrm{dist}((x,\xi),K)>\epsilon.
$$
\end{lemma}
\begin{proof}
As in \cite{MR1872698}*{Proof of Prop. 3.2.5} one then has
\begin{align*}
 (T_{\mu,h} u)(x,\xi) = \left( \frac{2 \pi h \sqrt{\mu}}{1+ \mu} \right) \alpha_h^2 \int
 	e^{\frac{\rmi}{h} (x \xi - y \eta)} e^{\frac{\rmi}{h} \frac{\mu x + y}{1+\mu} (\eta - \xi) } 
	e^{-\frac{\mu}{1 + \mu}\frac{(x-y)^2}{2h}} e^{-\frac{1}{1 + \mu}\frac{(\xi-\eta)^2}{2h}} (T_h u_h)(y,\eta) \der y \der \eta.
\end{align*}
Given $\epsilon>0$ let $K_\epsilon = \{(x,\xi) \mid \mathrm{dist}((x,\xi),K) < \epsilon\}$ and assume $(x,\xi) \notin K_\epsilon$. 
To estimate $(1+ |x|^2 + |\xi|^2)^N I_2(x,\xi)T_{\mu,h} u_h(x,\xi)$ we use this integral representation and split the result 
 into two parts $I_1(x,\xi)$ and $I_2(x,\xi)$. One is obtained from integrating over the ball $B_{\frac{\epsilon}{2}}(x,\xi)$ of radius $\frac{\epsilon}{2}$ centered at $(x,\xi)$, and the other is obtained by integrating over the complement of that ball. Then 
$I_2(x,\xi)$ satisfies the bound
\begin{align*}
 \| (1+ & |x|^2 + |\xi|^2)^N I_2(x,\xi) \|  \\ & \leq \left( \frac{2 \pi h \sqrt{\mu}}{1+ \mu} \right) \alpha_h^2 \int\limits_{|(x,\xi)-(y,\eta)|>\frac{\epsilon}{2}}  (1+|x|^2 + |\xi|^2)^N 
	e^{-\frac{\mu}{1 + \mu}\frac{(x-y)^2}{2h}} e^{-\frac{1}{1 + \mu}\frac{(\xi-\eta)^2}{2h}}\| (T_h u_h)(y,\eta)\| \der y \der \eta \\
\leq & \left( \frac{2 \pi h \sqrt{\mu}}{1+ \mu} \right) \alpha_h^2 e^{-\frac{\sigma}{2 h}\frac{\epsilon^2}{4}} \int  (1+|x|^2 + |\xi|^2)^N \|
	e^{-\sigma\frac{(x-y)^2}{2}} e^{-\sigma\frac{(\xi-\eta)^2}{2}} (T_h u_h)(y,\eta)\| \der y \der \eta 
\end{align*}
for any $h \in (0,1]$. Here $\sigma>0$ is chosen so that $\sigma<\mathrm{min}(\frac{1}{4}\frac{\mu}{1 + \mu},\frac{1}{4}\frac{1}{1 + \mu})$. The right hand side of this is exponentially small, as the integral is polynomially bounded. Recall that convolution with a Gaussian kernel is a continuous map on the space of Schwartz functions and $T_h u_h$ is a polynomially bounded family of Schwartz functions on $\R^{d} \times \R^d$.
The first term $I_1(x,\xi)$ is bounded as follows
\begin{align*}
 \| (1+  & |x|^2 + |\xi|^2)^N I_1(x,\xi) \|   \leq \left( \frac{2 \pi h \sqrt{\mu}}{1+ \mu} \right) \alpha_h^2 \int\limits_{|(x,\xi)-(y,\eta)|<\frac{\epsilon}{2}}  (1+|x|^2 + |\xi|^2)^N 
	\| (T_h u_h)(y,\eta)\| \der y \der \eta \\ & \leq
	\left( \frac{2 \pi h \sqrt{\mu}}{1+ \mu} \right) \alpha_h^2 C_d \epsilon^{2d} \sup_{(y,\eta) \notin K_\frac{\epsilon}{2}} (1+|\xi| + |\eta| + \epsilon)^{2N} 
	\| (T_h u_h)(y,\eta)\|.
\end{align*}
We have used that since $(x,\xi) \notin K_\epsilon$ the integration domain is contained in the complement of $K_\frac{\epsilon}{2}$.
The right hand side is exponentially small by assumption.
 \end{proof}
 
 \begin{proposition} \label{multmic}
  Assume $(u_h)$ and $(v_h)$ are in $\mathcal{S}_h(\R^d)$. If $K,K' \subset \R^d \times \R^d$, $(u_h)$ is uniformly microsupported in $K$, and  $(v_h)$ is uniformly microsupported in $K'$, then $(u_h \cdot v_h)$ is uniformly microsupported in
  $$
   K''=\{(x,\xi) + (x,\eta) \mid (x,\xi) \in K, (x,\eta) \in K'\}. 
  $$
 \end{proposition}
 \begin{proof}
 We can write $e^{-\rmi \xi x} T_h(u)(x,\xi) = (\pi h)^{-\frac{d}{4}} (\mathcal{F}_h)_{y \to \xi}(e^{-\frac{(x-y)^2}{2h}} u(y))$. We then 
  have for the semi-classical Fourier transform $\mathcal{F}_h(u_h) * \mathcal{F}_h(v_h) = (2 \pi h)^\frac{d}{2} \mathcal{F}_h(u_h v_h)$. We can therefore write the FBI-transform of a product as follows
  $$
   T_h( u_h v_h)(x,\xi) = (4 \pi h)^{-\frac{d}{4}} \int_{\R^d} \left((T_{\frac{h}{2}} u_h)(x,\xi-\eta)  \right) \left((T_{\frac{h}{2}} v_h)(x,\eta)  \right) \der \eta.
  $$
  Now note that if $(x,\xi)$ is in the complement of $K''_{\epsilon}$ and $\eta$ is arbitrary, 
  then either $(x,\eta)$ is in the complement of $K'_{\frac{\epsilon}{2}}$ or $(x,\xi - \eta)$ is in the complement of $K_{\frac{\epsilon}{2}}$. 
  The statement then follows immediately from Petree's inequality
  $$
    (1 + |\xi|^2 )^N \leq 2^N (1+|\xi - \eta|^2)^N(1+|\eta|^2)^N
  $$
  and using the fact that the individual factors in the integral are polynomially bounded in Schwartz space, and Lemma \ref{hilfe}.
   \end{proof}

\begin{proposition}
 Assume that $K \subset \R^d$ is compact and $u \in \mathcal{S}(\R^d)$ is analytic in a neighborhood of $K$. Then, for any $\epsilon>0$ we have that $u$ is microlocally uniformly exponentially small in $K \times \{\xi \in \R^d \mid \|\xi\|> \epsilon\} \subset \R^d \times \R^d$.
\end{proposition}
\begin{proof}
 The above follows from Iagolnitzer's characterisation of analytic functions by their essential support, namely \cite{MR0367657}*{Lemma 1,p. 122}. It is proved by a simple contour deformation argument. 
\end{proof}

Recall that if $F: \R^d \to \R^{d'}$ is a smooth map, then the pull-back $F^* u_h$ is defined by $(F^* u_h)(x) = u_h(F(x))$. If $(u_h) \in \mathcal{S}_h(\R^{d'})$ the pull-back will in general not be in  $\mathcal{S}_h(\R^d)$ unless decay conditions on all partial derivatives of the map are imposed. One can however multiply by any compactly supported smooth cut-off function $\chi$ to obtain a map  
$$\mathcal{S}_h(\R^d) \to  \mathcal{S}_h(\R^{d'}): u_h \mapsto \chi \cdot  (F^* u_h).$$

\begin{proposition} \label{cochart}
  Assume $(u_h) \in \mathcal{S}_h(\R^{d'})$ is microlocally uniformly exponentially small in $K \subset \R^{d'} \times \R^{d'}$. Assume that $F: \R^d \to \R^{d'}$ is a smooth map that is real analytic on the open subset $Q \subset \R^d$. Let $\chi \in C^\infty_0(\R^d)$ be a cut-off function that equals one on $Q$.
  Then $\chi \cdot(F^* u_h)$ is microlocally uniformly exponentially small in any set $K'$ with the following properties
 \begin{itemize}
   \item the projection $\mathrm{pr}_1(K')=\{ x \in \R^d \mid (x,\xi) \in K'\}$ of $K'$ to $\R^d$ is compact and contained in $Q$.
   \item we have $K'_\epsilon \subseteq \{ (x,(\der F(x))^T \eta) \in \R^{d} \times \R^{d} \mid  (F(x),\eta) \in K\}$ for some $\epsilon>0$.
 \end{itemize}
\end{proposition}
\begin{proof}
 Using \eqref{csrep} we write 
 $$
  u_h(F(x')) = \int_{\R^{d'}} \int_{\R^{d'}}(2 \pi h)^{-\frac{d}{2}} (T_h u_h)(y,\eta) \phi_{y,\eta,h}(F(x')) \der y \der \eta.
 $$
 Therefore, $T_h(\chi F^* u_h)(x,\xi) = \int \int K_h(x,\xi,y,\eta) (T_h u_h)(y,\eta) \der y \der \eta$, where
 \begin{align*}
  K_h(x,\xi,y,\eta) = 2^{-d} (\pi h)^{-\frac{3d}{4}}\int_{\R^d} \chi(x') e^{-\frac{(x-x')^2}{2h}} e^{-\frac{(y-F(x'))^2}{2h}} 
  e^{\frac{\rmi}{h} (x-x')\xi} e^{\frac{\rmi}{h} (F(x')-y)\eta} \der x'.
 \end{align*}

Repeated integration by parts gives
\begin{align}
 (1+ |\xi|^2)^k & K_h(x,\xi,y,\eta)=\nonumber \\ &\int_{\R^d} q_k(x,x',y,h,\eta) e^{-\frac{(x-x')^2}{2h}} e^{-\frac{(y-F(x'))^2}{2h}} 
  e^{\frac{\rmi}{h} (x-x')\xi} e^{\frac{\rmi}{h} (F(x')-y)\eta} \der x', \label{oscint}
\end{align}
where $q_k$ is a polynomial expression in $x$, $x'$, $y$, $h^{-1}$, and the partial derivatives of $F(x')$ and $\chi(x')$, as well as in $\eta$. The order of the polynomial in $\eta$ is at most $2k$.
Moreover, $q_k$ vanishes outside the support of $\chi$ and extends analytically to a complex neighborhood of $Q$ in the $x'$-variable.
For every $\ell >0$ there exists $r, C>0$ such that 
$$
\int\limits_{\supp(\chi)} | q_k(x,x',y,h,\eta) | e^{-\frac{(x-x')^2}{2h}} e^{-\frac{(y-F(x'))^2}{2h}}  \der x' \leq C h^{-r} (1+|x|^2+|y|^2)^{-\ell}
(1+|\eta|^2)^k,
$$
for $h \in (0,1]$,
where we have used that integration is over a compact set. We therefore get
\begin{equation} \label{ssestimate}
 (1+|x|^2+|y|^2)^{k}(1+ |\xi|^2)^k K_h(x,\xi,y,\eta) \leq (1+ |\eta|^2)^k C_{k} h^{-r_k}.
\end{equation}

We split the integral
\begin{align*}
 T_h(\chi F^* u_h)(x,\xi) = \int_K K_h(x,\xi,y,\eta) (T_h u_h)(y,\eta) \der y \der \eta \\+ \int_{\R^{2d}\setminus K} K_h(x,\xi,y,\eta) (T_h u_h)(y,\eta) \der y \der \eta = I_1(x,\xi) + I_2(x,\xi).
\end{align*}

Since $T_h(u_h)$ is uniformly exponentially small in $K$ the estimate \eqref{ssestimate} shows that 
$$ 
 \| (1 + x^2 + \xi^2)^k I_1(x,\xi) \| \leq C_k e^{-\delta h^{-1}}.
$$
We therefore only have to establish a bound of this form on $I_2(x,\xi)$ for $(x,\xi) \in K'$. 
To do this we will need to find an estimate for $K_h(x,\xi,y,\eta)$ in the set $K' \times K^c$, where $K^c$ is the complement of $K$ in $\R^{2 d'}$. In fact in this set the kernel is exponentially small.
To see this define a phase function
$$
 \varphi(x,x',y,\xi,\eta)= \rmi \frac{1}{2}(x-x')^2 +\rmi \frac{1}{2} (y-F(x'))^2 +(x-x')\xi + (F(x')-y)\eta.
$$
and write the integral \eqref{oscint} as
\begin{align*}
 \int\limits_{\supp(\chi)} q_k(x,x',y,h,\eta) e^{\frac{\rmi}{h} \phi(x,x',y,\xi,\eta)} \der x' =  \int\limits_{\mathrm{pr}_1(K')_{\tilde \epsilon}} q_k(x,x',y,h,\eta) e^{\frac{\rmi}{h} \varphi(x,x',y,\xi,\eta)} \der x' \\+  \int\limits_{\supp(\chi) \setminus \mathrm{pr}_1(K')_{\tilde \epsilon}} q_k(x,x',y,h,\eta) e^{\frac{\rmi}{h} \varphi(x,x',y,\xi,\eta)} \der x' = k_1(x,\xi,y,\eta,h) + k_2(x,\xi,y,\eta,h).
\end{align*}

Here $\tilde\epsilon>0$ is chosen small enough so that $\mathrm{pr}_1(K')_{\tilde \epsilon}$ is contained in $Q$. In particular this means on $\mathrm{pr}_1(K')_{\tilde \epsilon}$ the phase $\varphi(x,x',y,\xi,\eta)$ and the amplitude $q(x,x',y,h)$ extend to complex analytic functions of $x'$ in an open neighborhood of $\mathrm{pr}_1(K')_{\tilde \epsilon}$ in $\C^d$.

If $(x,\xi) \in K'$ the point $x$ therefore has distance at least ${\tilde \epsilon}>0$ from $\supp(\chi) \setminus \mathrm{pr}_1(K')_{\tilde \epsilon}$. The second integral therefore satisfies
$$
 | k_2 (x,\xi,y,\eta,h) | \leq C_k e^{-\frac{{\tilde \epsilon}^2}{4h}}
$$

For the first integral we use analytic non-stationary phase. For a point to be stationary with respect to $x'$ we need $-\xi + (\der F(x'))^T \eta=0$ and $(x-x') + \der F(x')(y-F(x'))=0$. If we assume that $(x,\xi,y,\eta) \in K' \times K^c$, then we have either $|y-F(x)|>\frac{\epsilon}{2}$ or $|\xi - (\der F(x))^T \eta | >\frac{\epsilon}{2}$.
Fix a small $\epsilon'>0$. Let $U$ be the set of points $x'$ such that $\Im(\varphi(x,x',y,\xi,\eta))< (\epsilon')^2$ for some
$(x,\xi,y,\eta) \in K' \times K^c$. For such a point we must have $|x-x'|< \epsilon'$ and $|y- F(x')|< \epsilon'$. 
This implies that $|y-F(x)| < c \epsilon'$ for some fixed $c>0$ that depends only on $F$. In particular we can choose $\epsilon'$ small enough so that $|y-F(x)|<\frac{\epsilon}{2}$ and $|\xi - (\der F(x'))^T \eta | >\frac{\epsilon}{4}$. This means $U$ does not contain any critical points for the phase function $\phi$. All the assumption of analytic non-stationary phase are fulfilled (see for example \cite{MR699623} or \cite{gevrey}*{Prop. 1.1} for a nice presentation) when integrating over $U$.
When integrating over the complement we have exponential decay. Summarising, this means our assumptions imply that
we have the uniform estimate
$$
 | \int\limits_{\supp(\chi)} q_k(x,x',y,h,\eta) e^{\frac{\rmi}{h} \phi(x,x',y,\xi,\eta)} \der x' | \leq C_{q_k} e^{-\delta h^{-1}}
$$
for some $\delta>0$, independent of $q_k$. This in turn gives
$$
 (1+ |\xi|^2)^k  (1+ |\eta|^2)^{-k} K_h(x,\xi,y,\eta) \leq C_k e^{-\delta h^{-1}}
$$
on $K' \times K^c$. Since $T_h u_h$ is polynomially bounded in $\mathcal{S}(\R^d \times \R^d)$  this finally gives the required estimate 
$$ 
 \| (1 + x^2 + \xi^2)^k I_2(x,\xi) \| \leq \tilde C_k e^{-\delta h^{-1}}.
$$
and the proof is finished.
\end{proof}

If $(u_h) \in C^\infty_{0,h}(\R^d)$ then there exists a compact set in which all the $u_h$ are supported.
If $F$ is an analytic diffeomorphism than $(F^* u_h)$ is also in $C^\infty_{0,h}(\R^d)$ and the above 
Proposition simply says that if $(u_h)$ is uniformly microsupported in $K \subset \R^d \times \R^d = T^* \R^d$, then $(F^* u_h)$ is uniformly microsupported in $(\der \Phi)^* K \subset T^* \R^d$. In other words being uniformly microsupported in a subset of the cotangent bundle is a well defined concept on an analytic manifold.

 \end{appendix}

\end{document}